\documentclass[lettersize,journal]{IEEEtran}
\usepackage{amssymb}
\usepackage{amsmath}
\usepackage{amsthm}
\usepackage{graphicx}

\usepackage[linesnumbered,ruled]{algorithm2e}
\usepackage{booktabs}
\usepackage{cite}
\usepackage{bm}
\usepackage{float}
\usepackage{multirow}
\usepackage{color}
\usepackage{subfigure}
\usepackage{caption}
\usepackage{epstopdf}
\usepackage{hyperref}
\usepackage{fixltx2e}
\usepackage{longtable}
\usepackage{diagbox}
\usepackage{changepage}
\usepackage{comment}
\usepackage{cases}
\usepackage[short]{newoptidef}
\usepackage{makecell}
\usepackage{mathabx}
\theoremstyle{definition}

\newtheorem{mechanism}{Mechanism}
\newtheorem{definition}{Definition}
\usepackage{algorithmic}
\theoremstyle{remark}

\theoremstyle{plain}
\newtheorem{theorem}{Theorem}
\newtheorem{lemma}{Lemma}

\newtheorem{assumption}{Assumption}

\renewcommand{\algorithmicrequire}{\textbf{Input:}}
\renewcommand{\algorithmicensure}{\textbf{Output:}}
\begin{document}

\title{Cached Model-as-a-Resource: Provisioning Large Language Model Agents for Edge Intelligence in Space-air-ground Integrated Networks}
% \begin{comment}

\author{Minrui Xu, Dusit Niyato, \emph{Fellow, IEEE}, Hongliang Zhang, Jiawen Kang*, Zehui Xiong,\\ Shiwen Mao, \emph{Fellow, IEEE}, and Zhu Han, \emph{Fellow, IEEE}
\thanks{M.~Xu and D.~Niyato are with the School of Computer Science and Engineering, Nanyang Technological University, Singapore 639798, Singapore (e-mail: minrui001@e.ntu.edu.sg; dniyato@ntu.edu.sg). H. Zhang is with School of Electronics, Peking University, Beijing 100871, China (e-mail: hongliang.zhang92@gmail.com). J.~Kang is with the School of Automation, Guangdong University of Technology, China (e-mail: kavinkang@gdut.edu.cn). Z.~Xiong is with the Pillar of Information Systems Technology and Design, Singapore University of Technology and Design, Singapore 487372, Singapore (e-mail: zehui\_xiong@sutd.edu.sg). S.~Mao is with the Department of Electrical and Computer Engineering, Auburn University, Auburn, AL 36849-5201 USA (email: smao@ieee.org). Z. Han is with the Department of Electrical and Computer Engineering at the University of Houston, Houston, TX 77004 USA, and also with the Department of Computer Science and Engineering, Kyung Hee University, Seoul, South Korea, 446-701 (e-mail: hanzhu22@gmail.com). (\textit{*Corresponding author: Jiawen Kang})}
}

\maketitle
\begin{abstract}
Edge intelligence in space-air-ground integrated networks (SAGINs) can enable worldwide network coverage beyond geographical limitations for users to access ubiquitous and low-latency intelligence services. Facing global coverage and complex environments in SAGINs, edge intelligence can provision approximate large language models (LLMs) agents for users via edge servers at ground base stations (BSs) or cloud data centers relayed by satellites. As LLMs with billions of parameters are pre-trained on vast datasets, LLM agents have few-shot learning capabilities, e.g., chain-of-thought (CoT) prompting for complex tasks, which raises a new trade-off between resource consumption and performance in SAGINs.
In this paper, we propose a joint caching and inference framework for edge intelligence to provision sustainable and ubiquitous LLM agents in SAGINs.
We introduce ``cached model-as-a-resource" for offering LLMs with limited context windows and propose a novel optimization framework, i.e., joint model caching and inference, to utilize cached model resources for provisioning LLM agent services along with communication, computing, and storage resources. We design ``age of thought" (AoT) considering the CoT prompting of LLMs, and propose a least AoT cached model replacement algorithm for optimizing the provisioning cost. We propose a deep Q-network-based modified second-bid (DQMSB) auction to incentivize network operators, which can enhance allocation efficiency by 23\% while guaranteeing strategy-proofness and free from adverse selection.
% to measure the relevance and coherence of intermediate thoughts 
% to motivate network operators to compete for opportunities to provision high-quality LLM agents, 
% However, the provisioning quantity and quality of AI agents necessitates the simultaneous availability of conventional resources, i.e., computing, communication, storage, and novel AI resources, i.e., LLMs. 
% As foundation models with billions of parameters pre-trained on vast datasets, LLMs demonstrate few-shot learning capabilities in tackling various downstream tasks, including in-context learning for unseen tasks and chain-of-thought prompting for complex tasks. Nevertheless, these capabilities are constrained by the size of context windows which depends on the model architecture of LLMs.
% Therefore, LLMs can be depleted during utilization similar to traditional computing, communication, and storage resources. In this regard, the allocation of LLMs should be investigated including the metrics for evaluation and the algorithms for execution. Finally, we propose a new optimization framework called model caching to manage LLMs to optimize model efficiency in edge intelligence. Considering the multi-step reasoning of LLMs, we propose a metric called age-of-thought to evaluate the quality of thoughts in context windows and propose a least-thought model caching algorithm to determine the loading and eviction of LLMs.
\end{abstract}


\section{Introduction}

Space-air-ground integrated networks (SAGINs) can extend network coverage for users to access real-time edge intelligence services, such as image recognition and data analysis, with global connectivity~\cite{liu2018space}. Beyond terrestrial communication systems, which are limited by network capacity and the coverage of ground-based stations (BSs), satellites can provide seamless connectivity for users in environmentally harsh areas such as oceans and mountains, by acting as relays during provisioning computing services from cloud data centers~\cite{tang2021computation}. For instance, implementing edge intelligence in SAGINs can enable various smart ocean activities~\cite{wei2021hybrid}, including real-time monitoring and analysis for offshore aquaculture and AI assistant services requested by onboard passengers, crew, and fishermen. The advance of large language models (LLMs)~\cite{min2023recent, xu2023sparks, xu2024large} dramatically enhance the capabilities of edge intelligence in SAGINs, allowing AI agents based on LLMs, i.e., LLM agents to tackle unseen and complex reasoning tasks with various data modalities~\cite{xi2023rise}. Moreover, edge intelligence in SAGINs can provision low-latency and privacy-preserving LLM agent services~\cite{xi2023rise} at edge servers of ground BSs or at cloud data centers relayed by satellites, which can act as autonomous life or work assistants for human users.

Based on LLMs with billions of parameters and pre-trained on Internet-scale datasets, LLMs agents can perform few-shot learning~\cite{brown2020language, yang2024llm}, including in-context learning (ICL) for unseen tasks chain-of-thought (CoT) prompting for complex reasoning tasks and role-playing under specific instructions. Although the training and inference of LLMs require enormous computing and data resources, users with low-end mobile devices in SAGINs can request LLM agent services from edge servers at ground BSs directly or from cloud data centers using satellites/ground BSs as relays. Therefore, provisioning LLM agent services in SAGINs with heterogeneous computing resources can reduce service latency and preserve user privacy~\cite{jiang2023approaching, hong2022protecting, xu2024unleashing}. Different from the traditional implementation of edge intelligence~\cite{zhou2019edge}, the resource-constraint edge servers at ground BSs cannot load every LLM to provision all LLM services simultaneously~\cite{xu2023sparks}. Meanwhile, the capabilities of LLMs to perform few-shot learning are limited by the size of context windows, which are determined by the model architecture of LLMs~\cite{packer2023memgpt}. During the provisioning of LLM agent services, inference operations accumulate tokens in context windows to elicit few-shot learning capabilities for better performance. After the accumulated inference tokens exceed the context windows~\cite{yang2023longqlora}, the output of LLMs might no longer be effective and the performance of LLM agents decreases significantly.

As the context windows of running LLMs can be depleted during provisioning LLM agent services, the cached models at edge servers of ground BSs should be considered as an unexplored type of resource analogous to conventional communication, computing, and storage resources~\cite{xu2023sparks}. In this regard, network operators aiming to minimize the provisioning cost of LLM agent services need to schedule the running LLMs at ground BSs based not only on conventional hardware constraints but also on the effective management of running LLMs. Similar to content caching, model caching is an optimization framework for implementing edge intelligence that aims to reduce service latency and resource consumption of network operators. Nevertheless, facilitating model caching for LLMs should consider their few-shot learning capabilities~\cite{brown2020language}, which affect resource consumption and inference performance at each inference step. Furthermore, the valuations of the opportunities to provision LLM agent services for network operators are positively correlated. However, as simple relays between users and cloud data centers, satellites have asymmetric information and cannot receive sufficient cost/performance feedback to measure their valuation compared to ground BSs, which might lead to adverse selection~\cite{arnosti2016adverse}, i.e., inefficient network operator allocation, for real-time mechanisms.

In this paper, we propose a joint model caching and inference framework for edge intelligence to provision sustainable and ubiquitous LLM agent services to users in SAGINs, where LLMs can be run at edge servers at ground BSs and cloud data centers via the relay of satellites/ground BSs. In this framework, ground BSs can leverage edge servers close to users for provisioning low-latency and privacy-preserving LLM agent services. Moreover, satellites can relay LLM agent service requests to cloud data centers for users in remote areas, e.g., mountains and oceans. To improve quality of services (QoS) for LLM agents running on ground BSs, we introduce a new optimization framework for serving LLMs in edge intelligence, i.e., joint model caching and inference. In this framework, cached LLMs can be regarded as a type of fundamental resource in provisioning LLM agent services. Leveraging the few-shot learning capability of LLMs, e.g., chain-of-thought prompting, we design a specific metric for the model caching algorithm, namely, age of thought (AoT), to measure the relevance and coherence of intermediate thoughts during CoT inference. Based on AoT, we propose the least AoT cache replacement algorithm, which determines the eviction of cached models with the least AoT, to reduce the total provisioning cost of network operators.
To eliminate adverse selection against satellites, we leverage modified second-bid (MSB) auctions~\cite{arnosti2016adverse} and deep reinforcement learning (DRL)~\cite{mnih2015human} to design a deep Q-network-based MSB (DQMSB) auction that can select the optimal price scaling factor in MSB to improve the marking efficiency in network operator allocation while guaranteeing the DQMSB is strategy-proof and free from adverse selection.

Our main contributions can be summarized as follows.
\begin{itemize}
    \item We formulate a novel optimization framework for edge intelligence, i.e., the joint model caching and inference framework, to provision sustainable and ubiquitous LLM agents for users with satellites and ground BSs in SAGINs.
\item In this framework, for the first time, we propose the concept of ``cached model-as-a-resource" to implement edge intelligence, where cached models are regarded as a type of resource similar to conventional communication, computing, and storage resources, at edge servers at ground BSs and cloud datacenters in SAGINs.
\item We formulate the LLM agent provisioning problem for ground BSs to minimize total system cost under resource and coverage constraints. To tackle this problem effectively, we design a novel least AoT model caching algorithm to schedule the loading and eviction of LLMs based on the metric of AoT, evaluating the relevance and coherence of intermediate thoughts in context windows.
\item To maximize the revenue of network operators in provisioning high-quality LLM agent services, we propose the DQMSB auction, which can guarantee free of  adverse selection and fully strategy-proof, by using DRL to select the optimal pricing scaling factor.
\end{itemize}

The remaining sections of this paper are organized as follows. In Section~\ref{sec:related}, we present a review of related work. In Section~\ref{sec:system}, we describe the system model for provisioning LLM agents in SAGINs. In Section~\ref{sec:problem}, we formulate the problem, propose the model caching algorithm, and design the market. In Section~\ref{sec:auction}, we propose the DQMSB auction. In Section~\ref{sec:results}, we present the simulation experiments. Finally, we conclude this paper in Section~\ref{sec:conclusions}.

% 第一段：大背景， SAGINs里面用LLM实现 边缘智能有什么好处 介绍LLMs的各种功能

% 第二段：LLM的各种功能都依赖于他的few-shot learning的能力，着重介绍ICL 和 COT， 这些能力能够通过消耗更多的tokens提升LLMs的推理能力， 同时这些消耗的tokens 会有一个上限， which is the size of context windows。如果LLM的推理tokens 超过了这个上线，那么 LLMs 的性能就会急剧下降。

% 第三段：引出我们的的研究问题， 在为LLMs分配资源的时候，我们不仅需要考虑传统的计算， 通信，存储资源， 我们还需要对缓存的模型和其中的上下文进行有效分配，以实现scalable 的性能。介绍 LLMs in wireless Edge networks， 首先是优势，low latency and privacy-preserving. 然后就是由于资源有限，不能够将全部LLM 模型同时保存早edge server中，这里从两方面论证，首先是LLMs很大很多能干很多事，其次是边缘服务器不能同时存全部模型。

% 第四段，在边缘智能中，用户的请求多样且动态，（举两个atttactive的例子，latency and accuracy），被加载到边缘服务器中的模型还还有context windows 的限制，有概率会被被动退出, 最后引出模型缓存，但是模型缓存就是怎么样设计指标衡量已缓存模型的表现和设计模型缓存算法。 在空天地网络中提供LLMs 服务有两种方式， 第一种是使用ground base station计算LLMs， 第二种是使用 satellites + cloud data centers 进行 LLMs 的计算的中继。但是为两种节点混合的网络提供激励机制有一些困难，由于satellite 和 ground base station 之间的信息不对称会造成 adversial selection，最后体现在market inefficiency (low revenue)。miss value

% In this paper, we propose an intelligent auction-based incentive mechanism for provisioning LLMs in SAGINs.

% Our contribution can be summarized as follows.
% 1. We propose a framework for provisioning LLMs in SAGINs
% 2. We introduce a new concept called cached model-as-a-resource during the provisioning of LLMs.
% 3. We introduce a new metric and propose a new caching algorithm.
% 4. We propose an intelligent auction based on the modified second-price auction and deep reinforcement learning. Furthermore, we prove that this auction is ...

\section{Related Works}\label{sec:related}

\subsection{Edge Intelligence in Space-air-ground Integrated Networks}
Provisioning AI services in SAGINs can significantly enhance the intelligent configuration and control of SAGINs to adapt to their environment, improving various performance metrics such as latency, energy usage, bandwidth, and real-time adaptability~\cite{liu2018space}. Xu \textit{et al.} in \cite{xu2023ai} introduce a cloud-edge aggregated artificial intelligence architecture that leverages the on-orbit lightweight 5G core and edge computing platform provided by the Tiansuan constellation. 
% Specifically, aiming to enhance data transmission reliability, enable timely data processing, and improve the utilization of satellite resources, the Tiansuan constellation enables SAGIN to perform intelligent data processing and forwarding by supporting satellite edge computing with a multi-controller model to effectively manage the computing load and task distribution. 
For mission-critical 6G services, Hou \textit{et al.} in \cite{hou2022edge} propose a three-layer architecture in SAGINs for ultra-reliable and low-latency edge intelligence that includes unikernel-based ultra-lightweight virtualization and microservice-based paradigms for prompt response and improved reliability. Considering the time-varying characteristics of content sources and the dynamic demands of users, Qin \textit{et al.} in \cite{qin2022content} propose a content service-oriented resource allocation algorithm that aims to achieve a stable matching based on users' preferences for SAGINs.
% Edge intelligence facilitates the gathering, processing, and analyzing of data locally, which is crucial for maintaining the privacy and efficiency of data exchange in SAGINs.

% focusing on caching-related work, including content caching and model caching

\subsection{Large Language Models for Edge Intelligence}

In literature, LLMs are an essential part of next-generation edge intelligence systems, which have been leveraged to design, analyze, and optimize edge intelligence~\cite{chen2023big, shen2024large}. For instance, Du \textit{et al.} in \cite{du2023power} investigate the potential of LLMs as a valuable tool for FPGA-based wireless system development. Furthermore, Cui et al. in \cite{cui2023llmind} introduce LLMind, an AI framework that integrates LLMs with domain-specific AI modules and IoT devices for executing complex tasks. In multi-agent systems for 6G communications, Jiang \textit{et al.} \cite{jiang2023large} demonstrated the effectiveness of LLMs in collaborative data retrieval, planning, and reflection through a semantic communication case study. Considering the issues that traditional deep offloading architectures are facing several issues, including heterogeneous constraints, partial perception, uncertain generalization, and lack of traceability, Dong \textit{et al.} \cite{dong2023lambo} propose an LLM-based offloading framework that utilizes LLMs for offloading decisions, addressing issues like heterogeneous constraints and uncertain generalization. Nevertheless, the execution of LLMs usually requires enormous computing resources, which is infeasible for edge environments. Therefore, considering efficient training and inference architecture in 6G networks, Lin \textit{et al.} in \cite{lin2023pushing} explore feasible techniques such as split learning/inference, parameter-efficient fine-tuning, quantization, and parameter-sharing inference for pushing LLMs to the edge.
% Here, we will have two parts, the first one is to utilize LLMs to construct EI and the second part is to highlight the provisioning of personal assistant services based on LLMs with EI.

\subsection{Auction Design for SAGINs}

Auctions are efficient and effective methods for real-time network resource allocation in SAGINs~\cite{zhang2021privacy,niyato2020auction, du2018auction}. 
% Furthermore, they analyze the performance of the designed auction-based traffic offloading mechanism and derive the unique optimal bidding strategies for different beam groups of the satellite to achieve the symmetric Bayesian equilibrium as well as the expected utility of the mobile network operator. 
In civil aircraft augmented SAGINs, Chen \textit{et al.} \cite{chen2020service} propose a truthful double auction for device-to-device (D2D) communications and a reverse auction mechanism for spectrum sharing. 
% These auctions can maximize overall profits in civil aircraft augmented SAGINs, where efficient and economic use of spectrum is critical for maintaining communication in challenging environments and ensuring the coexistence of terrestrial networks and sky access platforms. 
For lightweight blockchain-based SAGINs, Yang \textit{et al.} in \cite{yang2023lightweight} propose a secure sequential Vickrey auction mechanism to ensure secure and reliable spectrum sharing within the SAGINs.
% The sequential Vickrey auction is a type of sealed-bid auction where bidders submit their bids without knowing the bids of the other participants, and the highest bidder wins but pays the price bid by the second-highest bidder. 
However, existing auction designs for incentivizing network operators in SAGINs consider the information of satellites, ground BSs, and UAVs to be symmetric for evaluating their valuation of network resources. Nevertheless, the asymmetric information among network operators in SAGINs can {\color{black}cause} efficiency loss during resource allocation, especially when the values of network operators are positively correlated. Therefore, in this paper, we propose DQMSB that can mitigate adverse selection and increase market efficiency for SAGINs with DQN-based price scaling factor selection.

% Auction, game, contract in SAGINs. Highlight our learning-based auction in SAGINs

\section{System Model}\label{sec:system}

\begin{figure}[t]
    \centering
    \includegraphics[width=1\linewidth]{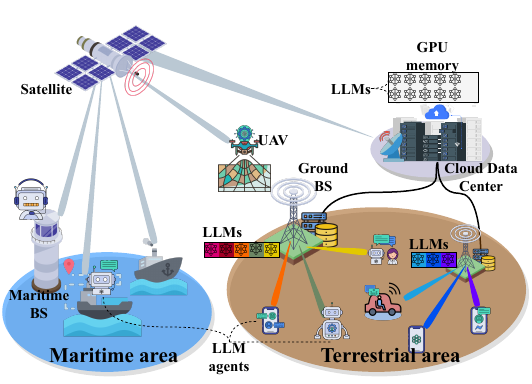}
    \caption{Joint caching and inference framework for provisioning large language model (LLM) agents in SAGINs.}
    \label{fig:system}
\end{figure}

For edge intelligence in SAGINs, each group of users would like to utilize LLM agents based on one or several LLMs. In the system, each group of users can use one LLM agent as their active assistant because their attention is limited depending on their preferences and current tasks. In SAGINs, a group of users needs to select network providers, including satellites and ground BSs, to access LLM agent services. 
% Since the number of ground BSs is significantly greater than the number of satellites, we consider that the opportunity to serve LLM agents for users can be shared by one satellite or multiple ground BSs. 
As shown in Fig. \ref{fig:system}, the system consists of $N+1$ network operators, including one or several Low Earth Orbit (LEO) satellites in orbit and multiple ground BSs equipped with edge servers, all connected to the cloud data center via backhaul links. The set of network operators is represented by $\mathcal{N} = \{0, 1, \ldots, N\}$, where the LEO satellite is represented by $0$ and the set of BSs is represented by $\{1, \ldots, N\}$. The edge servers at ground BSs can execute LLM agent services for users while the rest of the services can be offloaded to cloud data centers with the relay of satellites or ground BSs. We use the set $\mathcal{I}=\{1,2,\ldots, I\}$ to denote the available LLM agent services based on the set of LLMs $\mathcal{M}=\{1, \ldots, M\}$. As LLMs are capable of performing multiple downstream tasks in LLM agent services simultaneously, it is considered that the number of LLM agent services is far greater than the number of LLMs~\cite{xi2023rise}, i.e., $I \gg M$. In the group of users $\mathcal{U}_n$ covered by network operator $n$, $R_{n}^t = \{R_{n,i,m}^t| i \in \mathcal{I}, m \in \mathcal{M}\}$ is used to represent the number of requests generated by LLM agent service $i$ to execute LLM $m$ for its specific functions, including planning, memory, tool-using, and embodied actions. Initially, the size of input data of LLM agent service $i$ can be denoted as $d_i$. Additionally, the configuration of LLM $m$ consists of the amount of runtime GPU memory, which is proportion to model size $s_m$, the computation required per token $e_m$, and the size of context window $w_m$. 

\begin{figure*}
    \centering
    \includegraphics[width=1\linewidth]{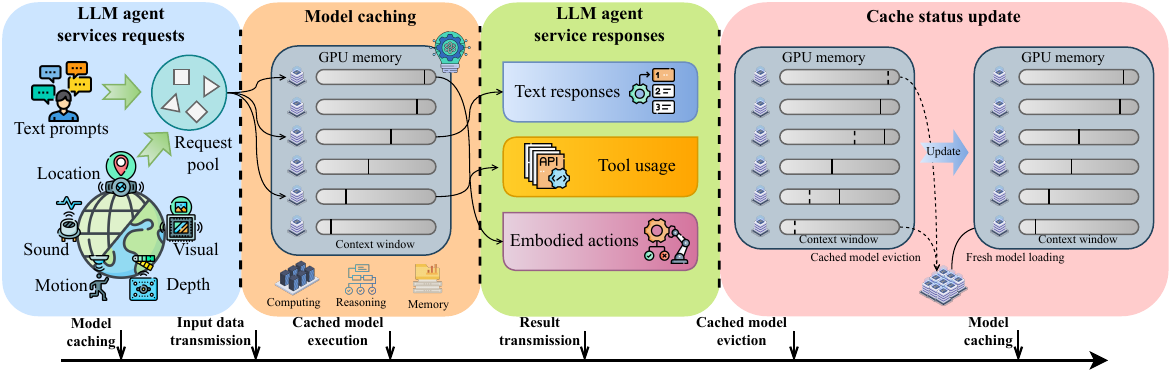}
    \caption{The workflow of the joint caching and inference framework for provisioning LLM agents with cached models.}
    \label{fig:workflow}
\end{figure*}

\subsection{Coverage Time Model}

The coverage time model for LEO satellite networks addresses the dynamic positioning of satellites concerning users. Unlike terrestrial networks, whose infrastructure remains stationary, LEO satellites exhibit constant motion, necessitating that users establish connections based on specific geometric metrics~\cite{tang2021computation}. These metrics include altitude $l$ of the LEO satellite orbit above the mobile user, the Earth's radius $E$, and the slant distance $s$ from users to the LEO satellites. The elevation angle $\theta^e$, delineating the line of sight between a mobile user and an LEO satellite, is determined by $\theta^e = \arccos{\left(\frac{E+l}{s}\right)}\cdot \sin\theta^g$,
where $\theta^g$ represents the geocentric angle covering the LEO satellite's service area, calculated as $\theta^g = \arccos{\left(\frac{E}{E + l}\right)}\cdot \cos\theta^e - \theta^e$.
Let $v^S$ denote the velocity of LEO satellite $0$. The maximal communication duration $T_0^S$ between a mobile user and the LEO satellite is given by
\begin{equation}
T_0^S = \frac{L}{v^S},
\end{equation}
where $L = 2 \theta^g (E + l)$ is the arc length over which communication with the LEO satellite is available for users.

%Let $h$ represent the distance between the users and the LEO satellite orbit, $E$ denote the radius of the earth, $s$ expresses the distance between the mobile user and the LEO satellite, and $\theta$ be the elevation angle between the mobile user and the LEO satellite, which can be obtained by

\subsection{Communication Model}

To facilitate interaction with LLM agents, a group of users covered by network operator $n$, denoted by $\mathcal{U}_n$, can access services through network operator $n$ for data transmission. These users share the same spectral resources, resulting in mutual interference among them~\cite{wei2021hybrid}. The channel power gain from mobile user $u \in \mathcal{U}_n$ to LEO satellite $0$, accounting for large-scale fading and shadowed-Rician fading, is represented by $g_{u, 0}$~\cite{deng2020ultra}. Similarly, $g_{u, n}$ represents the channel power gain from mobile user $u$ to ground BS $n=1,\ldots, N$, incorporating large-scale fading and Rayleigh fading for terrestrial communications. The bandwidth allocated by satellite $0$ and ground BSs $n = 1, 2, \ldots, N$, is denoted as $B_0$ and $B_n$, respectively. Consequently, the uplink transmission rate for user $u \in \mathcal{U}_n$ to transmit input data of LLM agent services to network operators is given by
\begin{equation}
r_{u,n} = B_n \log_2 \left(1+\frac{g_{u, n}p_u}{\sum_{j\in \mathcal{U}_n\backslash \{u\}}g_{j,n} p_j + \sigma^2}\right),
\end{equation}
where $p_u$ is the transmit power of user $u$ and $\sigma^2$ is the power of the additive white Gaussian noise (AWGN). The satellite serves as an intermediary in providing LLM agent services between mobile users and cloud data centers via the satellite backbone network~\cite{cheng2019space}, with the transmission rate denoted by $r^C_0$. Moreover, ground BSs $n = 1, 2, \ldots, N$ connect to cloud data centers through the terrestrial core network with a fixed transmission rate $r^C_n$.

\subsection{Model Caching for LLM Agent Services}

To facilitate the provisioning of LLM agent services in SAGINs, we introduce a joint model caching and inference framework that enables edge servers located at ground BSs to cache LLMs and offload requests, optimizing the utilization of edge computing resources to provision LLM agent services for users, as shown in Fig.~\ref{fig:workflow}. Specifically, ground BSs $n = 1, \ldots, N$ are tasked with determining local caching and offloading strategies. Here, $a_{n, i, m}^t \in \{0,1\}$ represents the binary variable that indicates whether model $m$ for service $i$ is cached at ground BS $n$ during time slot $t$, and $b_{n, i, m}^t \in [0,1]$ signifies the continuous variable reflecting the proportion of model $m$ for service $i$ being executed at ground BS $n$ at time slot $t$. Let $\mathbf{a}^t_n = \{a^t_{n,1,1}, \ldots, a^t_{n,I,M}\}$ encapsulate the model caching decisions at ground BS $n$, with $\mathbf{a}^t = \{\mathbf{a}^t_1, \ldots, \mathbf{a}^t_N\}$ aggregating these decisions across all network operators. Furthermore, the request offloading decision for ground BS $n$ is denoted by $\mathbf{b}^t_n = \{b^t_{n,1,1}, \ldots, b^t_{n, I, M}\}$, while $\mathbf{b}^t = \{\mathbf{b}^t_1, \ldots, \mathbf{b}^t_N\}$ represents the collective offloading decisions of ground BSs $n=1,\ldots, N$. To extend the coverage of ground communication systems, satellite $0$ with limited computing and energy resources acts as a relay between users and cloud data centers, whose model caching decisions are $\mathbf{a}^t_0 = \mathbf{0}$ and request offloading decisions are $\mathbf{b}^t_0 = \mathbf{0}$, i.e., all the LLM agent services are offloaded to cloud data centers for remote executing relayed by satellite $0$.

For edge intelligence in SAGINs, LLM agent services requested by users can be executed at edge servers at ground BSs when the required LLMs are cached into GPUs. Let $G_n$ denote the GPU computing capacity in terms of GPU memory of ground BS $n$. Then, the decision of  model caching $\mathbf{a}^t_n$ should satisfy the following constraint at time slot $t$ for ground BS $n = 1, \dots, N$, as
\begin{equation}
    \sum_{i\in \mathcal{I}} \sum_{m \in \mathcal{M}} a_{n, i, m}^t s_{m} \leq G_n,
\end{equation}
where $s_{m}$ is the running size of model $m$. This indicates that the edge servers cannot load all the LLMs into GPUs as the computing resources at edge servers are constrained.
After the models are loaded into the GPUs of edge servers, LLM agent services can be executed at the ground BSs. Therefore, the constraint of LLM agent service provisioned at ground BS $n = 1, \dots, N$, is represented as
\begin{equation}
    b_{n, i, m}^t \mathbf{1}(R_{n,i,m}^t>0) \leq a_{n, i, m}^t, \forall i\in\mathcal{I}, m\in\mathcal{M},
\end{equation}
where $\mathbf{1}(\cdot)$ is the indicator function and $\mathbf{1}(R_{n,i,m}^t>0)$ indicates that there are requests of LLM agent service $i$ for model $m$ at BS $n$ at time slot $t$.
Finally, the total computing power consumption of edge servers is constrained by the total computing capacity of GPUs at ground BS $n = 1, \dots, N$, which can be represented as
\begin{equation}
     \sum_{i\in \mathcal{I}} \sum_{m \in \mathcal{M}} e_m a_{n, i, m}^t (1 - b_{n, i, m}^t) R_{n, i,m}^t \leq E_n.
\end{equation}
Nevertheless, in cloud data centers, it could be assumed that there is no GPU memory constraint or computing capacity constraint for executing LLMs.
% \end{comment}
\subsection{Chain-of-Thought Inference Model}

To improve the relevance and coherence of LLM agents, LLMs can leverage CoT prompting to perform step-by-step reasoning before obtaining the final response~\cite{zhang2023igniting}. {\color{black} 
As an advanced inference approach to elicit the emerging abilities of LLMs, CoT prompting allows LLMs to generate a sequence of intermediate reasoning steps towards problem-solving or concluding, instead of attempting to solve the entire problem in merely zero-shot manner.} During inference of LLMs, given any task description prompt $d$, LLM $m$ can generate an answer by recursively predicting the sequence of next tokens from the learned distribution $p_m$ conditioned on the concatenation of $d$ and of the tokens sampled so far. {\color{black}For all sequences of messages in LLM agent service $i$, $D_i= \{d_{i,0}, \ldots, d_{i,k}\}$ of at most $w_m$ tokens, the $p_m(D_i)$ follows the general product rule of probability~\cite{jiang2023latent}, i.e.,
\begin{equation}
\begin{aligned}
    &p_m(D_i)  = p_m(d_{i,0}, \ldots, d_{i,k})\\
    & = p_m(d_{i,0})p_m(d_{i,1}|d_{i,0})\ldots p_m(d_{i,k}|d_{i,0},\ldots,d_{i,k-1}),
\end{aligned}
\end{equation}
which is a good approximation of the true distribution $\hat{q}(D_i)$.}

For eacg CoT prompt in LLM agent service $i$, {\color{black}LLM $m$ is provided with $c_i$ varying length CoT examples $E_i = \{e_{i,0}, \ldots, e_{i,k}\}$ and each thought $e_{i,k}$ in $E_i$ is a sequence of $k_i$ tokens representing one reasoning step.} Those examples are designed to aid the LLMs in producing correct answers via CoT generation and thus for service $i$, $E_i$ are generated with true intentions $\theta^\star$ and true context $c^\star$. To sever LLM agent service $i$, for the given $E_i$ and a task $d_{i,0}$, LLMs then generate $(d_{i, 1}, \ldots, d_{i, k})$ messages. {\color{black}To evaluate the performance of the approximation of LLMs, we have the following definition.}

\begin{definition}[$\epsilon$-ambiguity~\cite{jiang2023latent}]
    For CoT examples $E_i$ of LLM service $i$ generated based on true context $c^\star$ and true intention $\theta^\star$, the ambiguity of the chain $\epsilon(E_i)$ is defined as the complement of the likelihood of the context $c^\star$ and intentions $\theta^\star$ conditioned on CoT examples $E_i$, i.e., 
    \begin{equation}
        \hat{q}(c^\star, \theta^\star|E_i) = 1 - \epsilon(E_i).
    \end{equation}
\end{definition}
{\color{black}In addition to the ambiguous definition of LLMs, we define the quality of contexts in training datasets as follows.}

% \begin{proposition}
%     Consider that two $\epsilon$-ambiguity messages, i.e., $x_1$ and $x_2$, are independently generated under a common intention $\theta_*$, and their ambiguity levels are $\epsilon_1$ and $\epsilon_2$ respectively. If there are two independent messages are concatenated as a single composite message $(x_1, x_2)$, the level of ambiguity in determining the shared intention $\theta_*$ decreases to $\epsilon_1 \epsilon_2$.
% \end{proposition}

\begin{definition}\label{df1}
To account for the potentially non-uniform distribution of contexts in training datasets at network operators, we introduce a skewness parameter $\gamma_n(c^\star)$ for each network operator $n$, which is defined as
\begin{equation}
    \gamma_n(c^\star) = \sup_{c\in C}\frac{\hat{q}(c^\star)}{\hat{q}(c_n)},
\end{equation}
where $c_n$ is the context owned by network operator $n$.
\end{definition}

Ground BSs can collect the context during their provisioning of LLM agent services, which satisfies the preference of their local users. Therefore, we have the following assumption.

\begin{assumption}\label{context}
    {\color{black}The prior distribution associated with true contexts $c^\star$ is uniform.}
\end{assumption}

Based on the uniform context considered in Assumption~\ref{context}, $\gamma_n(c^\star) = 1$ guarantees small values or provides CoT examples that should have small enough ambiguity, so the model with high certainty could guess the true context $c^\star$ from them \cite{tutunov2023can}. Following Definition \ref{df1}, we can estimate the difference of ambiguity measurements between the learned distribution $p_m$ and the true distribution $\hat{q}$, conditioned on input messages $D_i$ of service $i$, as follows.

\begin{theorem}\label{th1}
    Considering a collection of $c_i$ varying length CoT examples, which are generated from the intention $\theta^\star$ with the optimal context $c^\star$ sampled from $q_m(c)$ that satisfies Assumption~\ref{context}. Furthermore, let $d_{i,0}$ be the input message or task sampled from $q(\cdot | \theta^\star_0)$, which is generated from $\theta^\star_0$ sampled from $q_m(\cdot|c^\star)$. Then, for any sequence of messages $D_i$, we have
    \begin{equation}
        |p_m (D_i |d_{i,0}, E_i) - \hat{q}(D_i |d_{i,0},c^\star)|\leq \eta \prod_{y=1}^{c_i} \frac{\epsilon(E_{i,y})}{1-\epsilon(E_{i,y})},
    \label{eq:th1}
    \end{equation}
    where $\eta = 2 \frac{\epsilon(d_{i,0})}{1-\epsilon(d_{i,0})}$ depends on the ambiguity of the input.
\end{theorem}

\begin{proof}
    We provide a simplified proof to show how to obtain this bound with the necessary steps and the complete proof can be found in\cite{tutunov2023can}. Starting from $p_m (D_i |d_{i,0}, E_i)$, we have
    \begin{equation}
        \begin{aligned}
            p_m (D_i |d_{i,0},& E_i)  = \frac{\hat{q} (D_i, E_i)}{\hat{q} (d_{i,0}, E_i)} \\
            & = \frac{\hat{q} (D_i, E_i, c^\star) + \sum_{c\neq c^\star}\hat{q} (D_i, E_i, c)}{\hat{q} (d_{i,0}, E_i, c^\star) + \sum_{c\neq c^\star}\hat{q} (d_{i,0}, E_i, c)}\\
            & = \frac{\hat{q} (D_i\backslash \{d_{i,0}\}, E_i, c^\star) + \frac{\sum_{c\neq c^\star}\hat{q} (D_i, E_i, c)}{\hat{q} (d_{i,0}, E_i, c^\star)}}{1 + \frac{\sum_{c\neq c^\star}\hat{q} (d_{i,0}, E_i, c)}{\hat{q} (d_{i,0}, E_i, c^\star)}}\\
            & = \frac{\hat{q} (D_i\backslash \{d_{i,0}\}, E_i, c^\star) + \Lambda}{1 + \Upsilon},
        \end{aligned}
    \end{equation}
where $\Lambda$ and $\Upsilon$ are given by
\begin{equation}
    \Lambda = \frac{\sum_{c\neq c^\star}\hat{q} (D_i, E_i, c)}{\hat{q} (d_{i,0}, E_i, c^\star)} \text{and }
    \Upsilon  = \frac{\sum_{c\neq c^\star}\hat{q} (d_{i,0}, E_i, c)}{\hat{q} (d_{i,0}, E_i, c^\star)}.
\end{equation}
By leveraging the definition of ambiguity measure for CoT example $E_{i,y}$, we can establish the following bounds on $\Lambda$ and $\Upsilon $ as
\begin{equation}
    \Lambda, \Upsilon  \leq \frac{\gamma^{c_i}_n(c^\star)\epsilon(d_{i,0})}{1 - \epsilon(d_{i,0})} \prod_{y=1}^{c_i}\frac{\epsilon(E_{i,y})}{1-\epsilon(E_{i,y})}.
\end{equation}
Finally, combining the above components, we have
\begin{equation}
    \begin{aligned}
        |p_m (D_i |d_{i,0}, E_i) &- \hat{q}(D_i |d_{i,0},c^\star)| \\
        & =  \frac{|\Lambda - \Upsilon  \hat{q}(D_i/ \{d_{i,0}\}|d_{i,0},c^\star) |}{1 + \Upsilon }\\
        & \leq |\Lambda + \Upsilon  \hat{q}(D_i/ \{d_{i,0}\}|d_{i,0},c^\star)|\\
        & \leq \Lambda + \Upsilon \\
        & \leq \eta \prod_{y=1}^{c_i}\frac{\epsilon(E_{i,y})}{1-\epsilon(E_{i,y})},
    \end{aligned}
\end{equation}
where $\eta = 2\frac{\gamma_n^{c_i}(c^\star)\epsilon(d_{i,0})}{1-\epsilon(d_{i,0})}$, following Assumption~\ref{context} and indicating that $\gamma_n(c^\star)=1$ and thus $\eta = 2 \frac{\epsilon(d_{i,0})}{1-\epsilon(d_{i,0})}$.
\end{proof}

Theorem \ref{th1} indicates that the LLM prompted with CoT example $E_i$ is capable of approximating the true natural language distribution equipped with true context and intentions.

\begin{assumption}\label{ass2}
    The CoT example $E_i$ generated from $\theta^\star$ with a context $c^\star\sim q(c)$ is bounded by the ambiguity measure, i.e.,
    \begin{equation}
        \epsilon(E_i) = \hat{q}(c^\star, \theta^\star|E_i) \leq \sigma,
    \end{equation}
    where $\sigma \in [0, \frac{1}{2}]$.
\end{assumption}
Assumption \ref{ass2} implies that when carefully selected, CoT examples $E_i$ and the true context $c^\star$ can be recovered from $E_i$ with reasonably high certainty, i.e., the probability the $c^\star$ is behind $E_i$ is strictly greater than on a half. Given such CoT examples, we can transform the bound in Theorem~\ref{th1} into a geometrical convergence rate with the number of examples growing large as
\begin{equation}\label{eq:convergence}
        |p_m(D_i |d_{i,0}, E_i) - \hat{q}(D_i|d_{i,0},c^\star)|\leq \eta \beta^{c_i},
\end{equation}
where the CoT gain is $\beta = \frac{\sigma}{1-\sigma}\in [0,1)$.
Those examples described in Assumption~\ref{ass2} should be carefully selected to guarantee low ambiguity requirements. {\color{black}In practice, however, it can be challenging to collect such chain-of-thought examples, as there can be an assumption that allows us to measure ambiguity for a given sequence of thoughts as below.}
\begin{assumption}\label{eq:ass3}
    For the CoT examples $E_i$ generated from true intentions $\theta^\star$ with the true context context $c^\star \sim q(c)$, the associated ambiguity measure $\epsilon(E_i)$ vanishes as the length of sequence grows large as
    \begin{equation}
        \lim_{l\rightarrow \infty} \epsilon(E_i) = 0.
    \end{equation}
\end{assumption}
Assumption \ref{eq:ass3} implies that uncertainty over true context $c^\star$ and true intentions $\theta^\star$ for a sequence of thoughts is diminishing when more of these thoughts are collected. Therefore, for long enough CoT examples, the asymptotic requirement is sufficient to guarantee a low ambiguity measure. Satisfying Assumption \ref{eq:ass3}, we have the following lemma.

\begin{lemma}\label{lemma1}
    Considering CoT examples $E_i$ for any fixed $\sigma \in [0, \frac{1}{2})$ there is a length threshold $k^\star_{i,\sigma} \in \mathbb{N}$. For any $k_i \geq k^\star_{i,\sigma}$, we have
    \begin{equation}
        \epsilon(E_i) \leq \sigma.
    \end{equation}
\end{lemma}
\begin{proof}
    By selecting $\sigma\in [0, \frac{1}{2})$, CoT examples $E_i$ following Assumption \ref{eq:ass3} have the approximation that $\lim_{l\rightarrow \infty} \epsilon(E_i) = 0$. Then, there exists $k_{i,\sigma}^\star\in \mathbb{N}$ such that for any $k_i\geq k^\star_{i,\sigma}$, the inequality $\epsilon(E_i) \leq \sigma$ holds.
\end{proof}

Based on Lemma \ref{lemma1}, the geometrical convergence rate in Eq.~(\ref{eq:convergence}) can be established when the LLM is prompted with CoT examples $E_i$ of sufficient length, i.e., $k_i \geq k^\star_{i,\sigma}$ following Assumption \ref{eq:ass3}. In contrast to the low ambiguity requirement, CoT examples with low asymptotic ambiguity can be more attainable. Therefore, during the step-by-step inference of LLMs, the original CoT examples $E_i$ satisfying Assumption~\ref{eq:ass3} can be split or divided into different lengths and sizes of thoughts $E_i'$ following the required threshold $\sigma \in [0, \frac{1}{2})$, which can be more refined reasoning steps with predefined ambiguity.

LLMs, such as GPT-3, can perform CoT prompting which indicates that they can learn from past CoT examples for complex tasks presented to them. 
The intermediate thoughts can be used to enhance the performance of LLM agents, as LLMs can use meta-gradient learning during interaction to fit them~\cite {dai2023can}. However, depending on the relevance and coherence of intermediate thoughts, few-shot learning may have favorable or unfavorable impacts on the model performance. Based on the caching decision $a_{n, i, m}^t$ and offloading decision $b_{n, i, m}$, the batch of requests executed as ground BS $n$ can be calculated as $\delta_{n, i, m}^t= a_{n, i, m}^t (1-b_{n, i, m}^t) R_{n, i,m}^t k_i$ for service $i$ and model $m$, where $k_i$ is the size of CoT examples for service $i$ which can be estimated via Lemma \ref{lemma1}. In general, the number of thoughts increases monotonically when the LLM is cached into the edge servers, which can be represented as
\begin{equation}\label{eq:contexttoken}
	K^t_{n,i,m} = \begin{cases}
		0,& t = 0, \\ 
		a_{n, i, m}^t(K^{t-1}_{n,i,m} + \delta_{n, i, m}^t), &\text{otherwise}.
	\end{cases}
\end{equation}
Similar to the definition of age of information (AoI), the AoT measures the freshness of intermediate thoughts within the cached LLMs for the current inference requests.  With a vanishing factor $\Delta_{i, m}^t$ of thoughts, the AoT is adjusted by the non-increasing age utility function, which is represented as
\begin{equation}\label{eq:contextvalue}
   \kappa^t_{n,i,m} = \begin{cases}
 0,& t = 0, \\ 
 a_{n, i, m}^t \{ \kappa^{t-1}_{n,i,m} +  \delta_{n, i, m}^t - \Delta_{i,m}^t\}^+, &\text{otherwise}.
\end{cases}
\end{equation}
% \begin{equation}
%     K_{i,m}^t =  \min\left(w_m, \{K_{i,m}^{t-1} +  R_{n,i,m}^t a_{n,i,m}^{t} b_{n,i,m}^{t} - \Delta_{i,m}^t\}^+\right),
% \end{equation}
According to the AoT, the weighted total of the number of examples in demonstrations may be used to determine the number of examples in context.

Based on Eq.~(\ref{eq:convergence}), the few-shot CoT reasoning performance $A_{i,m}$ of model $m$ in service $i$ can be defined as
\begin{equation}\label{eq:acc}
    A_{n, i,m}^t = \alpha_{i,m}\log(1/\beta_{i, m}^{\kappa^t_{n,i,m}}), % \log_2(1+{K_{i,m}^t}^{\alpha_m}),
\end{equation}
where $\alpha_{i,m}$ is the zero-shot accuracy of LLM $m$ for service $i$ and $\log(1/\beta_{i, m}^{\kappa^t_{n,i,m}})$ is the performance gain of generated CoT examples of LLM $m$ for service $i$. 

\section{Problem Formulation, Caching Algorithm Design, and Market Design}\label{sec:problem}

In this section, we first formulate the problem of provisioning LLM agents in SAGINs to maximize the quantity and quality of LLM agent services. To solve the formulated problem, we next propose a model caching algorithm for local cached model management for ground BSs. Furthermore, we design an LLM agent market for incentivizing network operators, i.e., satellites and ground BSs, to contribute their resources to execute LLM agents. 

During the provisioning of LLM agent services, network operators, i.e., satellite or terrestrial BSs, need to provide communication, computation, storage, and cached model resources to run the LLM agent services. In detail, the user's request and returned results need to be transmitted over the wireless channel, which costs bandwidth to complete the interaction between the user and the LLM agent deployed at the edge. The reasoning process of LLM consumes a lot of computing resources, especially when LLMs perform CoT reasoning; they need a large amount of intermediate computation to obtain the final high-quality results. Since LLM uses CoT reasoning to improve the quality of model inference, they need substantial storage resources to store past CoT examples locally for demonstration. Finally, since the context window of LLM is limited, when the model inference is deployed beyond a certain number of tokens, the model needs to be refreshed, i.e., to evict the old model and load another fresh model. % Finally, due to the limited context window of LLM, the model needs to be refreshed, i.e., evict the depleted old model and reload a fresh model to continue LLM agent services.

\subsection{Cost Structure}

As mentioned above, the LLM agent services provisioned at ground BSs can be handled by edge servers or offloaded to cloud data centers via the core network. Based on the decisions of model caching and request offloading, the overall cost of providing LLM agent services, including the cost at the edge and the cost at the cloud, can be expressed as follows.

\subsubsection{Edge Inference Cost}
{\color{black}Specifically, the cost of edge inference includes the cost of switching in GPUs, the cost of transmitting data across edges, the cost of performing computations on edges, and the cost of model accuracy.
Based on decisions about model caching, each edge server must load models into the GPU memory before execution.} During the model loading process, there is a cost associated with switching between models in GPUs, which includes the latency of loading the model and the cost of wear and tear on the hardware~\cite{zhao2022edgeadaptor}. Therefore, the switching cost $l^s_n$ of ground BS $n$ to load and evict models can be calculated as
\begin{equation}
    l^{switch}_n(\mathbf{a}^t) = \sum_{i\in \mathcal{I}} \sum_{m\in\mathcal{M}} \lambda \mathbf{1}(a_{n,i,m}^{t} > a_{n,i,m}^{t-1}),
\end{equation}
where $\lambda$ denotes the coefficient for loading and evicting the model and $\mathbf{1}(\cdot)$ is the indicator function. When $a_{n,i,m}^{t} > a_{n,i,m}^{t-1}$, i.e., $a_{n,i,m}^{t}=1$ and $a_{n,i,m}^{t-1} = 0$, $\mathbf{1}(a_{n,i,m}^{t} > a_{n,i,m}^{t-1})$ indicates that the loading of an uncached model. Otherwise, there is no switching cost incurred at edge servers. 

When the requested models are cached into the GPU memory of edge servers, users communicate with the edge servers to request LLM agent services. Let $l_n^{trans}$ denote the transmission cost of input prompts and inference results. The transmission cost of ground BS $n$ can be calculated as
\begin{equation}
    l_n^{trans}(\mathbf{a}^t, \mathbf{b}^t) = \sum_{i\in\mathcal{I}} 
    \sum_{m\in\mathcal{M}} R_{n,i,m}^t\left(l_{n, i} +  \frac{d_i}{r_n^C} b_{n,i,m}^{t}\right),
\end{equation}
where $ l_{n, i} = d_i / \mathbb{E}_{u\in\mathcal{U}_n}[r_{u,n}]$ is the unit transmission cost per input and result for service $i$ to transmit the input data with size $d_i$ from users $\mathcal{U}_n$ to ground BS $n$.

Let $f_n$ denote the computing capacity of ground BS $n$. The execution of LLM agent services at ground BSs incurs inference latency, which is denoted as $l^{comp}_n$ for ground BS $n$. The edge computing cost can be calculated as
\begin{equation}
    l^{comp}_n(\mathbf{a}^t, \mathbf{b}^t) = \sum_{i\in\mathcal{I}} 
    \sum_{m\in\mathcal{M}} \delta_{n, i, m}^t \frac{e_m}{f_n},
\end{equation}
where $\delta_{n, i, m}^t= a_{n, i, m}^t (1-b_{n, i, m}^t) R_{n, i,m}^t k_i$ is the total computation token for a batch of LLM agent service $i$.
Finally, as ground BSs might not have sufficient resources for executing the best-matched model requested by LLM agents, which might introduce a performance gap, the requests processed by other LLMs with the equivalent function incur accuracy cost $l^{acc}_n$, which can be represented as
\begin{equation}
    l^{acc}_n(\mathbf{a}^t, \mathbf{b}^t) = \sum_{i\in\mathcal{I}} 
    \sum_{m\in\mathcal{M}} \bar{A}_{n, i,m}^t  R_{n,i,m}^t a_{n,i,m}^{t} (1 - b_{n,i,m}^{t}),
\end{equation}
where $\bar{A}_{n, i,m}^t = \frac{1-\alpha_{i,m}}{\kappa^t_{n,i,m} \log(1/\beta_{i, m})}$ is the unit accuracy cost following Eq.~(\ref{eq:acc}).
By sacrificing some accuracy of LLM agent services, the system can reduce the model missing rate.
Therefore, the total edge inference cost of ground BS $n$ is
\begin{equation}
\begin{aligned}
    L^t_n(\mathbf{a}^t, \mathbf{b}^t) = &l^{switch}_n(\mathbf{a}^t_n) + l_n^{trans}(\mathbf{a}^t_n, \mathbf{b}^t_n) \\&+ l^{comp}_n(\mathbf{a}^t_n, \mathbf{b}^t_n) + l^{acc}_n(\mathbf{a}^t_n, \mathbf{b}^t_n).
\end{aligned}
\end{equation}
The edge inference cost is jointly determined by the caching decisions and offloading decisions of ground BSs. Nevertheless, the missed or offloaded requests are executed by cloud data centers.

\subsubsection{Cloud Inference Cost}

The ground BSs are typically limited in resources and cannot serve all LLM agent service requests. There are two main reasons for this limitation. First, the ground BSs’ computing capacities may be insufficient to load many LLMs into the GPU memory. Second, the ground BSs’ energy capacity may not be enough to handle all requests. Therefore, some requests will be offloaded to cloud data centers for remote execution.

When the requested models are not available or the ground BS lacks sufficient resources, these user requests are transmitted to the cloud data center, which then allocates resources to serve them. According to~\cite{zhao2022edgeadaptor}, cloud data centers can provide serverless LLM agent services, charging users on a ``pay-as-you-go" basis. This means that users pay based on the number of requests rather than the specific resources occupied.
% When requests cannot be fulfilled by ground BSs due to no available cached model or when ground BSs lack resources, the uncompleted requests are offloaded to the cloud data center for remote execution. The cloud data center has abundant computing and energy resources to execute the models and return the inference results to the ground BSs. 
However, this cloud-based inference introduces additional latency due to data transmission in the core network, which is larger than the latency at ground BSs. Additionally, the accuracy cost of offloaded inference requests executed by the cloud data center is expected to be minimal, as the requests can be processed using the most accurate model with common CoT examples owned by the data center.

Based on the above analysis, we use $l_{n}^C$ to represent the total cost of offloading requests to the cloud data center for remote execution. Then, the total cloud computing cost at time slot $t$ is
\begin{equation}
    L^t_C(\mathbf{a}^t, \mathbf{b}^t) = \sum_{n\in\mathcal{N}} \sum_{i\in\mathcal{I}} \sum_{m\in\mathcal{M}} l_{0,m} b_{n,i,m}^t R_{n,i,m}^t,
\end{equation}
where $l_{0,m}$ is the unit process cost of model $m$ at cloud data centers.
Therefore, the total cost $L^{total}_n$ for provisioning LLM agent services at network operator $n$ can be calculated as
\begin{equation}
    L^{total}_n = \frac{1}{T}\sum_{t\in\mathcal{T}}\left(L^t_C + L_n^t\right).
\end{equation}

% \begin{comment}
\subsection{Problem Formulation of Ground Base Stations}

To improve the efficiency of mobile edge intelligence, we take into account both the cost of edge inference and cloud inference. This includes considering the switching cost, accuracy cost, transmission cost, and inference cost over a specific period $T$. For ground BSs $n = 1, \ldots, N$, the problem of providing LLM agent services is formulated as
\begin{mini!}|s|[2]<b>
    {\mathbf{a}^t_n, \mathbf{b}^t_n}{L^{total}_n\label{eq:obj}}{}{}
    \addConstraint{(3), (4), (5)}{\label{eq:con1}}
    \addConstraint{K_{n,i,m}^t\leq w_m, \forall i \in \mathcal{I}, \forall m\in\mathcal{M}}{\label{eq:con11}}
    \addConstraint{a_{n,i,m}^t\in}{\{0,1\}\label{eq:con2}}
    \addConstraint{b_{n,i,m}^t\in}{[0,1]\label{eq:con3}}.
\end{mini!}
Constraint (\ref{eq:con11}) indicates that the context tokens cannot exceed the size of context windows. 
To address the optimization problem described above, we need to overcome the challenge of time-coupled elements, such as GPU memory and CoT examples, as it takes into account both future request dynamics and historical CoT examples. Furthermore, the problem is a mixed-integer programming problem, which is known to be NP-hard. To solve the problem efficiently, we require low-complexity algorithms to determine decisions of model caching and request offloading.

\subsection{The Least Age-of-Thought Caching Algoirhtm}

To effectively serve LLMs for provisioning LLM agent services, we propose the least AoT algorithm based on the proposed AoT metric. When additional GPU memory is required for loading an uncached requested LLM, the least AoT algorithm counts the value of CoT examples, calculates them, and removes the cached LLM with the lowest AoT. Therefore, at each time slot $t$, the model caching decisions can be obtained by solving the maximization problem of the number of CoT examples for the cached models, which can be represented as
\begin{maxi!}|s|[2]<b>
    {\mathbf{a}^t}{\sum_{i\in\mathcal{I}}\sum_{m\in\mathcal{M}} \kappa^t_{n,i,m} \label{eq:tobj}}{}{}
    \addConstraint{\sum_{i\in \mathcal{I}} \sum_{m \in \mathcal{M}} a_{n,i,m}^t s_{m} \leq G_n, \forall n\in \mathcal{N}}{\label{eq:tcon1}}
    \addConstraint{\sum_{i\in \mathcal{I}} \sum_{m \in \mathcal{M}} (1-b_{n,i,m}^t) e_{m} \leq E_n, \forall n\in \mathcal{N}}{\label{eq:tcon1}}
    \addConstraint{a_{n,i,m}^t\in}{ \{0,1\}\label{eq:tcon2}}.
    % \addConstraint{b_{n,i,m}^t\in}{[0,1]\label{eq:con3}}.
\end{maxi!}
The available capacity of GPU memory $G_n^t$ of ground BS $n = 1, \ldots, N$ at time slot $t$ can be calculated as $G_n^t = G_n-\sum_{i\in \mathcal{I}} \sum_{m \in \mathcal{M}} a_{n,i,m}^t s_m$. 
% This optimization problem can be tackled with the complexity of $O(IM)$ with prior knowledge and statistical data.
This algorithm allows the least important LLM to have a higher chance for eviction in the current inference task. The complexity of the algorithm increases linearly as the number of models increases. Therefore, it works well with a large number of LLMs on ground BSs with limited GPU memory. Using more intermediate reasoning steps during inference makes the LLMs perform more accurately. Based on caching decisions $\mathbf{a}^t_n$ by solving the optimization problem in (\ref{eq:tobj}), offloading decisions $\mathbf{b}^t_n$ are obtained by solving the optimization problem in (\ref{eq:obj}). The detailed implementation of the Least AoT algorithm is provided in Algorithm~\ref{alg1}.

\begin{algorithm}[t]
		\caption{The Least Age-of-Thought Cached Model Replacement Algorithm}
		\label{alg1}
		\small
                \algorithmicrequire{ Model caching status $\mathbf{a}^{t-1}_n$, model context status $\kappa_n^t$, and LLM agent service requests $R^{t}_n$, GPU memory capacity $G_n$, GPU computing capacity $E_n$.}\\
                \algorithmicensure{ Model caching decision $\mathbf{a}^{t}_n$ and request offloading decision $\mathbf{b}^{t}_n$.}\\
                Initialize $\mathbf{a}^{t}_n = \mathbf{0}$, $\mathbf{b}^{t}_n = \mathbf{0}$, $G_n^t = G_n - \sum_{i\in\mathcal{I}}\sum_{m\in\mathcal{M}}a_{n,i,m}^{t-1}s_m$, and $E_n^t = 0$.\\
                \For{$R_{n,i,m}^t > 0 $ in $R_{n}^t$}{
                \If{$a_{n,i,m}^{t-1} = 1 $}{
                    $a_{n,i,m}^{t} \leftarrow 1$;\\
                }
                \ElseIf{ $G_n^t + s_m \leq G_n$}{
                    $a_{n,i,m}^{t} \leftarrow 1$;\\
                    $G_n^t \leftarrow G_n^t + s_m$;\\
                }
                \Else{ % 在缓存满的情况下将模型加载进去， 先推出模型，然后再加载
                \While{$G_n^t + s_m \geq G_n$}{
                $(\bar{i},\bar{m}) \rightarrow \arg\min_{(\bar{i},\bar{m})}\{\kappa_{n,\bar{i},\bar{m}}^t \in \kappa_n^t\}$;\\
                \If{$a_{n,\bar{i},\bar{m}}^{t-1} = 1$ and $a_{n,\bar{i},\bar{m}}^{t} = 0$}{
                $G_n^t = G_n^t - s_{\bar{m}}$;\\
                }
                }
                \If{$G_n^t + s_m \leq G_n$}{
                    $a_{n,i,m}^{t} \leftarrow 1$;\\
                    $G_n^t \leftarrow G_n^t + s_m$;\\
                }
                }
                \If{$a_{n,i,m}^{t}=1$ and $E_n^t + e_m a_{n,i,m}^{t}R_{n,i,m}^t \leq E_n$}{
                $b_{n,i,m}^{t} \leftarrow 1$;\\
                $E_n^t \leftarrow E_n^t + e_m a_{n,i,m}^{t}b_{n,i,m}^{t}R_{n,i,m}^t$;\\
                Update context status $\kappa_{n,i,m}^t$ following Eq.~(\ref{eq:contextvalue});\\
                }
                \If{$K_{n,i,m}^t > w_m$}{
                $a_{n,i,m}^{t} \leftarrow 0$;\\
                $b_{n,i,m}^{t} \leftarrow 0$.
                }
		}
  
\end{algorithm}

\subsection{Market Design}

To motivate LLM agent providers to construct and update LLM agents for users, we design an LLM agent market where sellers (LLM agent providers) can earn profits from provisioning LLM agent services, and bidders (network operators) are competing for provisioning LLM agent services to their users. We consider that the network operators are risk-neutral bidders in the market whose surpluses are positively correlated with each other based on the revelation principle~\cite{arnosti2016adverse}.

In SAGINs, users can obtain the services of a running LLM agent via a satellite or ground BS as their assistants to perform their local tasks. Each network operator $n \in \mathcal{N}$ has valuation $v_n = c_n m_n$ for the opportunity to serve the LLM agent, which is a production of the common value $c_n = \mathbb{E}_{t\in\mathcal{T}}[L^{total}_n - l^{acc}_n(\mathbf{a}^t, \mathbf{b}^t)]$ about physical resource consumption and the match performance gain $m_n = \mathbb{E}_{t\in\mathcal{T}}[\log(1/\beta_{i, m}^{\kappa^t_{n,i,m}})]$ obtained in Eq.~(\ref{eq:acc}). Specifically, the common value captures attributes of the resource consumption required to execute the LLMs, which depends on the communication, computing, and storage resources. Meanwhile, the match quality captures idiosyncratic components of cached models in network operators that affect the quality of LLM agents~\cite{arnosti2016adverse}. During the valuation of LLM agents, the common value is considered to be independent of match quality, i.e., the resource consumption of LLM agents is not relevant to the quality of LLM agents. We use $v_{(i)}$, $c_{(i)}$, and $m_{(i)}$ to denote the $i$-th highest valuation, common value, and match value factor, respectively.

In the market, a mechanism $\mathcal{M}(\mathbf{v}) = (\mathbf{z}, \mathbf{p})$ is required to map the privately held valuation $\mathbf{v}$ to allocation probabilities $\mathbf{z} = (z_0, z_1, \ldots, z_N)$ and payment $\mathbf{p} = (p_0, p_1, \ldots, p_N)$. The expected surplus, i.e., the realization of valuation, for satellite is $\mathbb{E}[v_0z_0(\mathbf{v})]$. Meanwhile, the surplus from the LLM agent allocated to the ground BSs is given by $\mathbb{E}[\sum_{n=1}^{N}v_nz_n(v)]$. To maximize the total surplus of network operators to provision LLM agents, the problem for the mechanism can be formulated as
\begin{maxi!}|s|[2]<b>
    {\mathcal{M}}{\mathbb{E}\left[ \sum_{n=0}^{N}v_iz_i(v)]\right]\label{eq:surplus}}{}{}
    % \addConstraint{h_{i,k}}{\leq c_i, \quad}{\forall i \in \mathcal{I}, k\in \mathcal{K}\label{con2}}{}{}
    \addConstraint{\sum_{t\in \mathcal{T}} \sum_{i\in\mathcal{I}} 
    \sum_{m\in\mathcal{M}} \frac{d_i R_{0,i,m}^t}{\mathbb{E}_{u\in\mathcal{U}_0}[r_{u,0} + r_0^C]}}{ \leq T_0^S\label{eq:con24}}{}{}
    \addConstraint{\sum_{n=0}^{N}z_{n}}{\leq 1\label{eq:con25}}{}{}{}
    \addConstraint{z_{n}}{\in \{0,1\}, \forall n \in \mathcal{N}\label{eq:con26}}{}{}.
\end{maxi!}
Constraint (\ref{eq:con24}) indicates that the provisioning time of satellites cannot exceed their coverage time. Constraints (\ref{eq:con25}) and (\ref{eq:con26}) indicate that there is one and only one network operator that can obtain the opportunity to run LLM agents.

% Buyers: There are two types of buyers in the market, the general buyer is satellite and the special buyers are the base stations. Specifically, the general buyer can estimate the real-time performance of the LLM agent service provided to the users while the specialist buyers can estimate the real-time performance based on user feedback.

% Auctioneer: The LLM agent platforms

% Seller and auctioneer: LLM agent provider or LLM agent. 

\section{The Deep Q-Network-based Modified Second-bid Auction} \label{sec:auction}

\subsection{Modified Section-price Auction}

In the LLM agent market, all network operators submit their bids $\mathbf{x} = (x_0, x_1, \ldots, x_N)$ to the auctioneers. Then, the auctioneer leverages MSB auction~\cite{arnosti2016adverse} to determine the winning bidder and payment, which can be formulated as follows.

\begin{mechanism}[Modified Second-bid Auction]
The MSB auction allocates the LLM agent to the highest BS if their bid exceeds the second-highest bid by a price scaling factor of $\rho$ or more and prices the LLM agent with the second-highest bid scaling with $\rho$. When no performance bidders win, the opportunity is allocated to the relaying satellite, which offloads LLM agent services to cloud data centers, with the contracted price $x_0$ is chosen to maximize its expected profit as $x_0 = \max_{x} \mathbb{E}[(v_0 - v_{(1)} \textbf{1}(v_{(1)\leq x_0}))]$~\cite{arnosti2016adverse}. Formally, the allocation rule and the pricing rule can be represented as follows.
\begin{itemize}
    \item Allocation rule: For ground BSs, namely, the performance bidders, $n = 1, \ldots, N$, the allocation probabilities $z_n \in \{0,1\}$ for the deterministic mechanism are determined by
    \begin{equation}\label{eq:allocation}
        z_n(\mathbf{x}) = \mathbf{1}(x_n > \rho \max\{\mathbf{x}_{-n}\}),
    \end{equation}
    for $\rho \geq 1$. Then, the allocation probability for the satellite can be calculated based on the allocation probabilities of ground BSs, as $z_0(\mathbf{x}) \leq 1 - \sum_{n=1}^N z_n(\mathbf{x})$.
    \item Pricing rule: If the winner is ground BS $n = 1, \ldots, N$, the winning ground BS is charged with the product of the price scaling factor $\rho$ the second highest bid, i.e., the payment $p_n$ can be calculated as
    \begin{equation}\label{eq:payment}
        p_n(\mathbf{x}) = z_n(\mathbf{x}) \cdot \rho \max\{\mathbf{x}_{-n}\}.
    \end{equation}
    Furthermore, the payment of the satellite depends on their contract price $x_0$, i.e., $p_0(\mathbf{x}) = z_0(\mathbf{x}) x_0$.
\end{itemize}
    
\end{mechanism}

\subsection{DQN-based Price Scaling Factor}

To leverage DRL to determine the price scaling factor, we formulate the process of MSB as a Markov decision process (MDP), consisting of states, actions, and rewards. At each time step, the DRL-based auctioneer observes the current state and selects an action, i.e., the price scaling factor from the feasible action space. Then, the auctioneer determines the winning bidder and the payment, and receives the total surplus as the reward. Formally, the MDP of MBS can be formulated as follows. During the auction process, multiple network operators submit their bids to the auctioneer. Therefore, the state space consists of the bidding information in MSB, i.e., $S_k \triangleq \{\mathbf{x}^k\}$ at decision slot $k$. To calculate the pricing rule and allocation rule, the auctioneer needs to determine the level of the price scaling factor, i.e., $a_k \in \mathcal{A} \subseteq  \mathbb{N}$. Then, the price scaling factor is calculated as $\rho = 10^{a_k/|\mathcal{A}|}$, where $|\mathcal{A}|$ is the size of action space. The reward is calculated as the total surplus achieved in the market, i.e., $r_k = \mathbb{E}\left[ \sum_{n=0}^{N}v_iz_i(v)]\right]$ calculated in Eq.~(\ref{eq:surplus}).

Based on the MDP of MBS, the objective of the DRL-based auctioneer is to optimize a non-linear function approximation with parameters $\phi$ as a Q-network to maximize the future discounted return $R_k = \sum_{k'=k}^K \gamma^{k'-k}r_{k'}$. In this regard, the auctioneer needs to learn the optimal action-value function as
\begin{equation}
    Q^\star(S, a) = \mathbb{E}_{S'}\left[r+\gamma \max_{a'}Q^\star(S',a')|S,a\right],
\end{equation}
where the action $a'$ is selected to maximize the expected value of $r+\gamma Q^\star(S',a')$. Therefore, the target for iteration $k$ can be calculated as $
    y_k = r_k + \gamma \max_{a^{k+1}} Q_{\phi'}(S^{k+1}, a^{k+1}),$
where $\phi'$ is the parameters of target network $Q_{\phi'}$. To minimize the performance gap between the current Q value and the target, the loss function can be defined as
\begin{equation}\label{eq:loss}
    L = \frac{1}{K} \sum_{k=1}^K\left(y_k - Q_\phi(S_k,a_k)\right)^2.
\end{equation}
For computational efficiency and performance stability, DQN leverages stochastic gradient descent to optimize the parameters $\phi$ on the loss calculated in Eq. (\ref{eq:loss}) and update the target network $\phi'$ in each period of iteration. Finally, the DQMSB auction is provided in Algorithm~\ref{alg2}.

\begin{algorithm}[t]
    \caption{The DQMSB Auction}
    \label{alg2}
    {\small
    \algorithmicrequire{ Bids $\mathbf{x}$;}\\
    \algorithmicensure{ Allocation probabilities and payments;}\\
        Initialize Q-function parameters $\phi$, target Q-function parameters $\phi'$, and replay buffer $\mathcal{B}$;\\
    \For{episode in $1, \dots,T$}{
        \For{iteration $k$ in $1, \dots, K$}{
            Receive the bids $\mathbf{x}$ from bidders and observe the state $S_k$;\\
            Determine the price scaling factor $\rho = 10^{a_k/|\mathcal{A}|}$, following $a_k = \arg\max_{a}Q_\phi(S_k, a)$;\\
            Calculate the winning probabilities $\mathbf{z}$ and payments $\mathbf{p}$ obtained from the allocation rule in Eq.~(\ref{eq:allocation}) and the pricing rule in Eq.~(\ref{eq:payment});\\
            Observe the next states $S_{k+1}$ and reward $r_k$;\\
            Store transition $(S_k, a_k, r_k, S_{k+1})$ to $\mathcal{B}$;\\
            Sample a mini-batch of  experiences $(S_k, A_k, r_k, S_{k+1})$ from $\mathcal{B}$;\\
            Calculate $y_k = r_k + \gamma \max_{a^{k+1}}Q_{\phi'}(S_{k+1}, a_{k+1})$ using target network $Q_{\phi'}$;\\
            Update $\phi$ by performing gradient descent on the loss calculated in Eq.~(\ref{eq:loss});\\
            Update target network $\phi'$.\\
        }
    }
}
\end{algorithm}

\subsection{Property Analysis}

For an auction, strategy-proofness means that participants cannot achieve a higher utility by altering their honest bids. Adverse-selection-free means that the presence of market externalities and asymmetric information is unrelated to bidders' valuations. 
% Then, under this mechanism, the factors of market externalities and asymmetric information are also unrelated to the allocation rules of the synchronization mechanisms. 
Therefore, it is important to note that the DQMSB auctions are fully strategy-proof and adverse-selection-free, as given in the following theorem.

\begin{theorem}
\label{the2}
    The DQMSB auction with the price scaling policy with fixed parameters $\bar{\phi}$ is anonymous, fully strategy-proof, and adverse-selection-free.
\end{theorem}
\begin{proof}
    To prove that the proposed DQMSB auction is anonymous, fully strategy-proof, and adverse-selection-free, the auction should be characterized by a critical payment function $\chi$ conditioned on $\bar{\phi}$ such, for any competing bids $\mathbf{x}_{-n}$, ground BS bidder $n = 1,\ldots, N$ wins if and only if its bid exceeds the critical payment $\chi(\mathbf{x}_{-n}; \bar{\phi}) = \rho^{a_k/|\mathcal{A}|} \max\{\mathbf{x}_{-n}^k\}$, where $a_k = \arg\max_{a} Q(\mathbf{x}_{-n}^k\cup \{\mathbb{E}[x_n]\}, a;\bar{\phi})$. Then, when the ground BS bidder $n$  is conditional on winning, it needs to pay the critical payment $\chi(\mathbf{x}_{-n}; \bar{\phi})$. As $\rho^{a_k/|\mathcal{A}|} \geq 1$, only the ground BS with the highest bid bidder can win, which can satisfy the condition that $\chi(\mathbf{x}_{-n}; \bar{\phi}) \geq \max\{\mathbf{x}_{-n}\}$.
    In addition, the critical payment function of DQMSB $\chi(\mathbf{x}_{-n}; \bar{\phi}) = \rho^{a_k/|\mathcal{A}|} \max\{\mathbf{\mathbf{x}_{-n}^k}\}$ satisfies 
    \begin{equation}
        \begin{aligned}
            \chi(\max\{\mathbf{x}_{-n}\}; \bar{\phi}) &= \rho^{a_k/|\mathcal{A}|} \max\{\max\{\mathbf{x}_{-n}\}\}\\
            & = \rho^{a_k/|\mathcal{A}|} \max\{\mathbf{x}_{-n}\}\\
            & =\chi(\textbf{x}_{-n}; \bar{\phi}).
        \end{aligned}
    \end{equation}
    Therefore, considering there are two bidders in the market with one value higher than $\chi(\mathbf{x}_{-n}; \bar{\phi})$ and the other one value $\max\{\mathbf{x}_{-n}\}$, which will cause $\chi(\max\{\mathbf{x}_{-n}\}; \bar{\phi}) \neq \chi(\mathbf{x}_{-n}; \bar{\phi})$, the DQMSB auction cannot satisfy false-name proof. Specifically, when  $\chi(\mathbf{x}_{-n}; \bar{\phi}) < \chi(\max\{\mathbf{x}_{-n}\}; \bar{\phi})$, the first bidder can submit a lower bid while maintaining the other bids in the set of bids and thus the auction is not winner false-name proof. Otherwise, the auction is not loser false-name proof when $\chi(\mathbf{x}_{-n}; \bar{\phi}) > \chi(\max\{\mathbf{x_{-n}}\}; \bar{\phi})$, where the losing bidder in the market can submit a higher bid compared with the winner’s bid while maintaining the other bids in the set of bids $\mathbf{x}_{-n}$.
    Furthermore, the critical payment function $\chi$ of the DQMSB auction is homogeneous of degree one, which indicates that the auction is adverse selection-free. Suppose that $\chi$ is not homogeneous of degree one, a bidder could manipulate the system by adjusting their bid in response to their private information $C\in\{1, c\}$, i.e., $\chi(\textbf{m}_{-n};\bar{\phi}) < \chi(c\textbf{m}_{-n};\bar{\phi})/c$, where $c\in \mathbb{R}_+, n\geq 2$, and $\mathbf{x}_{-n} \in \mathbb{R}^{n-1}_+$ and $C=1$, $z_n(C\mathbf{m}) = z_n(\mathbf{m}) = \mathbf{1}_{\{m_n > \chi(m_{-n};\bar{\phi})\}} = 1$, so $z_0(C\mathbf{m})=0$. {\color{black}When $C\neq 1$, it indicates that the bidder can change its bid from $C = c$ to influence the probability of winning the auction, i.e., $z_n(C\mathbf{m}) = z_n(c\mathbf{m}) = \mathbf{1}(cm_n > \chi(cm_{-n};\bar{\phi})) = 0$. }However, the ground BS bidder $n$ with the highest bid cannot win the auction, and thus the satellite bidder $0$ wins the auction, i.e., $z_0(C\mathbf{m})=1$.

    Based on the above analysis, the proposed DQMSB can guarantee the anonymous, fully strategy-proof, and adverse-selection-free with the deterministic $Q_{\bar{\phi}}$.
\end{proof}
% In the remainder of the paper, we use the notation $MSB_{a\sim\pi}$ to denote the modified second bid allocation rule with strategy $\pi$ in which $z_i(x) = 1_{x_n > \rho \max \{x_{-i}\}}$ and $z_0(x) = 1-\sum_{n=1}^N z_n(x)$.

Theorem \ref{the2} describes MSB auctions as robust auctions that are anonymous, deterministic, not prone to adverse selection, and fully strategy-proof. It is important to note that the authors based on other DRL algorithms may possess any three of these characteristics. By allowing non-anonymity, the value of $\rho$ determined by DRL algorithms can vary depending on the bids of the current bidders.

\section{Experimental Results}\label{sec:results}

In this section, we validate the joint model caching and inference framework while evaluating the performance of the proposed least AoT cached model replacement algorithm and the DQMSB auction.

% \usepackage{graphicx}
\subsection{Parameter Settings}

\begin{figure*}[t]
    \centering
    \subfigure[Average total cost versus time steps.]{\includegraphics[width=0.24\textwidth]{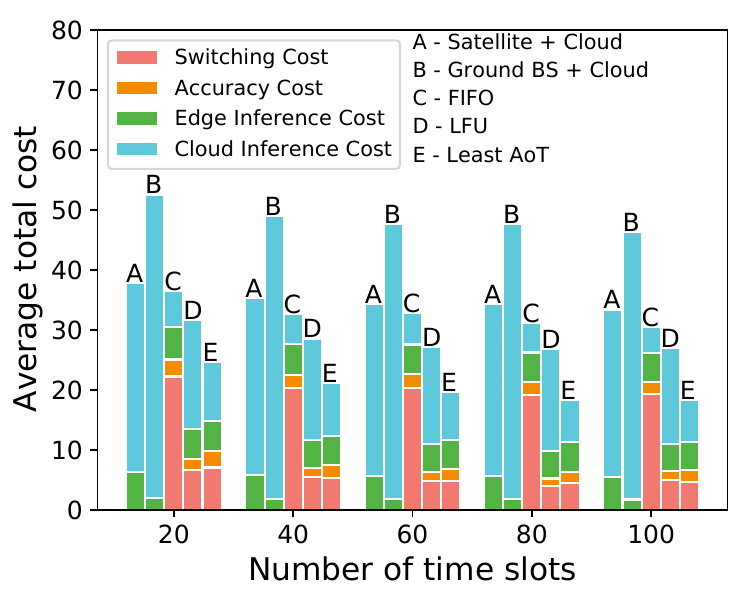}%
        \label{fig:SAGIN_time_step}}
    % \hfil
    \subfigure[Average total cost versus number of services.]{\includegraphics[width=0.24\textwidth]{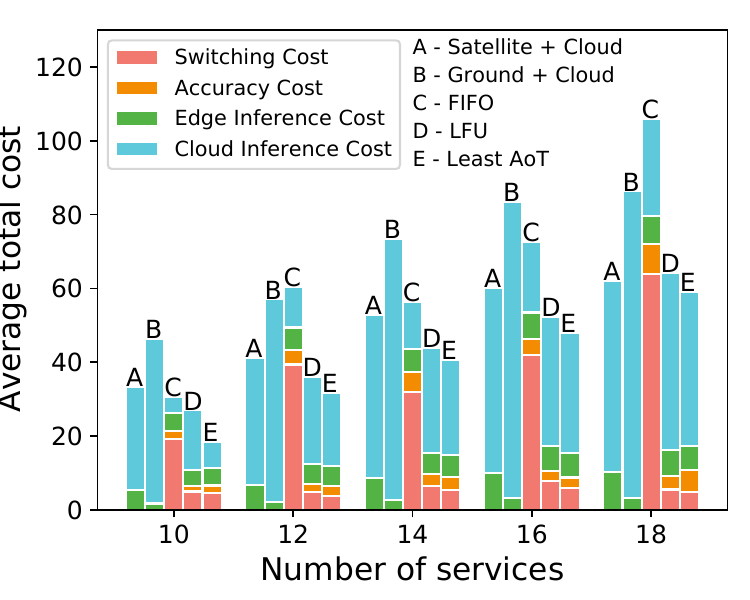}%
        \label{fig:SAGIN_services}}
    \subfigure[Average total cost versus number of GPUs.]{\includegraphics[width=0.24\textwidth]{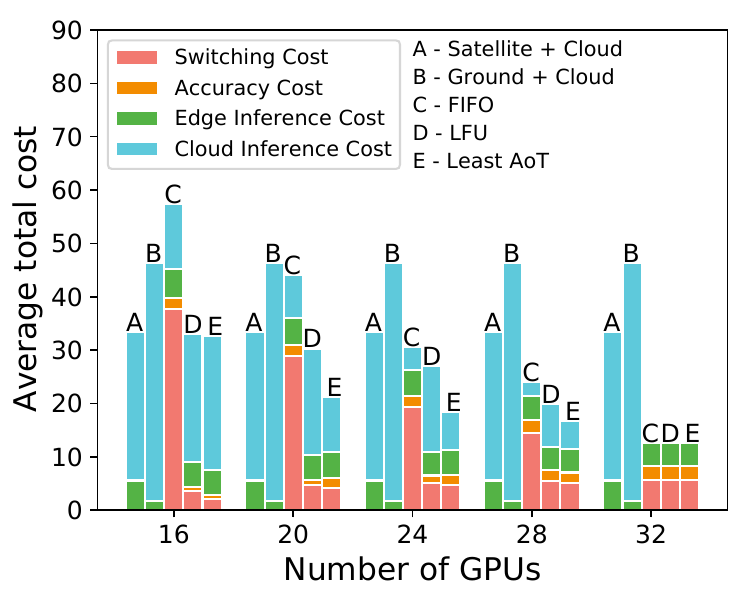}%
        \label{fig:SAGIN_GPUs}}
    \subfigure[Average total cost versus number of users.]{\includegraphics[width=0.24\textwidth]{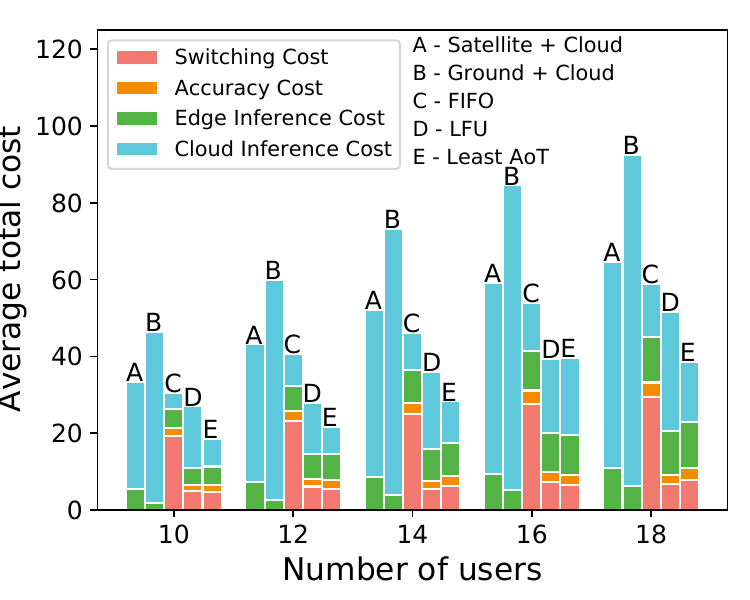}%
        \label{fig:SAGIN_users}}
    \caption{Performance of model caching algorithms under different system settings.}
    \label{fig:real_simulationrevenue}
\end{figure*}

We consider the SAGINs with one satellite and multiple ground BSs to provision LLM agent services to users. We evaluate the proposed algorithm within 100 time steps. For each GPU in the edge servers of ground BSs, the memory is 80 GB, energy efficiency is 810 GFLOPS/W, and energy capacity is 300 W. The default number of services is set to 10, the default number of GPUs at ground BSs is set to 24, and the default number of users is set to 10. We leverage ImageBind~\cite{girdhar2023imagebind} as the multimodal perception module of LLM agents, whose performance is 77.7\% for images, 50.0\% for videos, 63.4\% for infrared, 54.0\% for depth, 66.9\% for audio, 25.0\% for IMU. 
% TABLE~\ref{tab:imagebind}.
% \begin{table}[h]
% \centering
% \caption{Performance of the Imagebind-based multimodal perception module.}
% \label{tab:imagebind}
% \resizebox{0.48\textwidth}{!}{%
% \begin{tabular}{ccccccc}
% \hline
% Modality & Image & Video & Infrared & Depth & Audio & IMU \\ \hline
% Performance & 77.7 & 50.0 & 63.4 & 54.0 & 66.9 & 25.0 \\ \hline
% \end{tabular}%
% }
% \end{table}
We consider two types of LLMs for performing CoT reasoning, including the LLAMA-65B and GPT3-174B, whose context windows are 2k and 8k tokens, respectively. The total number of reasoning LLMs is set to 10. The default context vanishing factor $\Delta_n^t$ is set to 0.6. For each CoT reasoning example, the maximum size is set to 200 tokens. The transmit power of users is set to 0.2 W and the allocated bandwidth is set to 20 MHz. The size of input data of LLM agent services is uniformly selected from [100, 200] MB. The quantity of LLM service requests is generated from the Poisson point process depending on the number of users. Due to the ratio of existing ground BS and communication satellites, the edge access cost for satellites is set to 0.005 and 0.0001 for ground BSs. The cloud access cost is set to 0.04 for ground BSs and 0.025 for satellites. For the LLM agent service market, the number of BSs is set to 5 by default. The satellite-related settings follow~\cite{tang2021computation} and the DRL-related parameter settings follow~\cite{weng2022tianshou}.

\subsection{Performance Evaluation of the Least AoT Algorithm}

% \begin{figure}[t]
%     \centering
%     \includegraphics[width=0.8\linewidth]{Figs/SAGIN_time_step.pdf}
%     \caption{Total average cost versus time steps.}
%     \label{fig:timestep}
% \end{figure}
% \begin{figure}[t]
%     \centering
%     \includegraphics[width=0.8\linewidth]{Figs/SAGIN_services.pdf}
%     \caption{Total average cost versus number of services.}
%     \label{fig:services}
% \end{figure}
% \begin{figure}[t]
%     \centering
%     \includegraphics[width=0.8\linewidth]{Figs/SAGIN_GPUs.pdf}
%     \caption{Total average cost versus number of GPUs.}
%     \label{fig:GPUs}
% \end{figure}
% \begin{figure}[t]
%     \centering
%     \includegraphics[width=0.8\linewidth]{Figs/SAGIN_users.pdf}
%     \caption{Total average cost versus number of users.}
%     \label{fig:users}
% \end{figure}

\begin{figure}[t]
\vspace{-0.3cm}
    \centering
    \subfigure[Accuracy cost versus vanishing factor.]{\includegraphics[width=0.49\linewidth]{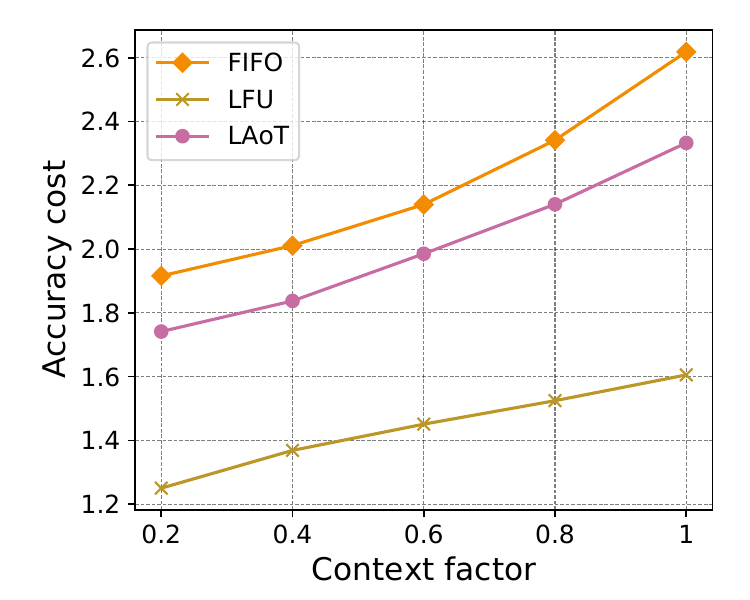}%
        \label{fig:accuracycost}}
    % \hfil
    \subfigure[Performance gain versus vanishing factor.]{\includegraphics[width=0.49\linewidth]{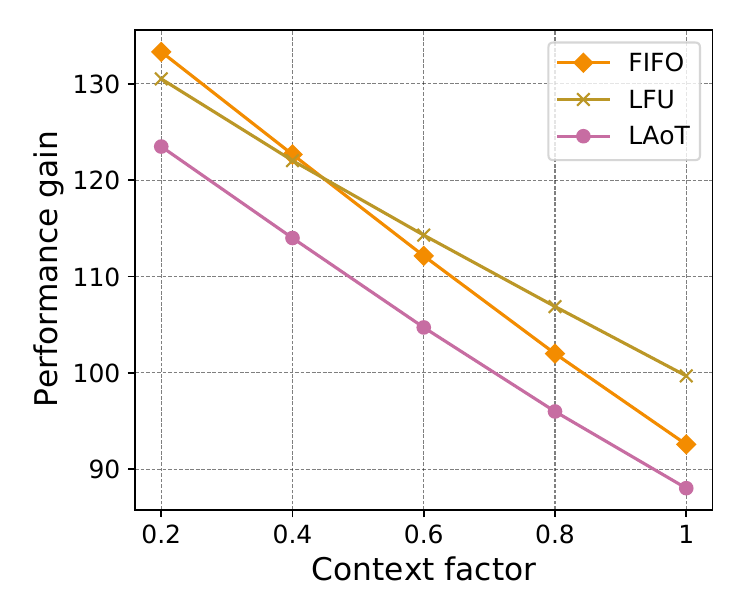}%
        \label{fig:performancegain}}
    \caption{Model performance under different vanishing factors.}
    \label{fig:real_simulationrevenue}
\end{figure}
During the performance evaluation of the proposed model caching algorithm, we leverage several traditional caching baselines for comparison, including first-in-first-out (FIFO) and least frequently used (LFU) algorithms.
For instance, Fig~\ref{fig:SAGIN_time_step} presents the average cost of different caching and offloading strategies (A-E) over a series of time slots. The costs are divided into four categories: switching cost, accuracy cost, edge inference cost, and cloud inference cost. For all strategies, there is a clear trend of decreasing average cost as the number of time slots increases. This indicates that over time, the systems may become more efficient, possibly due to improved caching optimization or better distribution of LLM agent tasks between the edge and the cloud. 
% Specifically, the least AoT algorithm consistently demonstrates the lowest average cost across all time slots, highlighting its efficiency. The overall trend suggests that the least AoT algorithm is the most effective in minimizing average cost, especially in terms of cloud inference cost, which represents a significant portion of the total cost. This implies that the least AoT algorithm is likely to keep relevant LLMs cached at the edge, reducing the need for expensive cloud computations. 
In addition, Fig.~\ref{fig:SAGIN_services} demonstrates the increase of the average total cost for different caching and offloading schemes as the number of services increases. This suggests that as the system has to handle a larger variety of services, the associated costs rise, possibly due to increased complexity and demand for resources. The least AoT algorithm demonstrates a consistent performance advantage, maintaining the lowest average total cost across different numbers of services. % It seems to mitigate this effect best, suggesting that its method of cache management is particularly effective at keeping relevant services available at the edge, thus reducing the need for expensive cloud inferences. The marginal increase in accuracy cost with the number of services for all schemes suggests that maintaining a high level of service accuracy becomes more challenging as service diversity grows. This could reflect the inherent difficulty in predicting the caching needs for a larger set of LLMs.

Furthermore, Fig. \ref{fig:SAGIN_GPUs} illustrates how the average total cost changes for various caching and offloading schemes with the number of GPUs utilized. As we can observe, more available GPUs tend to lower the cost, likely due to the improved computational efficiency and reduced processing time at the edge.  The relative stable switching and accuracy costs across different GPU counts suggest these costs are more dependent on the efficiency of the algorithm itself rather than on hardware resources. Meanwhile, edge inference cost reductions point to the benefits of local processing power, highlighting the importance of edge capabilities in managing LLM caching. Overall, the average total cost decreases for most schemes as GPU resources increase. % Still, the benefit is most pronounced for the least AoT algorithm, which consistently outperforms other strategies across varying levels of GPU availability.
Finally, Fig. \ref{fig:SAGIN_users} demonstrates an upward trend in average total cost for all schemes as the number of users increases. This suggests that the system's costs escalate with the growing user base, likely due to increased demand for LLM services, which intensifies the load on caching and computation resources. The increase in average total cost across all schemes with more users suggests that user demand has a direct impact on the system's resource utilization and cost efficiency. The least AoT algorithm demonstrates scalability by maintaining the lowest increase in cost, indicating its potential for cost-effective expansion as user numbers grow. The relative stability of the switching cost across varying user counts may imply that the action of switching between cached LLMs does not contribute significantly to cost variations. % On the other hand, the accuracy and inference costs, which increase with user numbers, suggest that these are the main drivers of cost in response to increased demand. It suggests that the least AoT algorithm, which prioritizes the AoT, is efficient not only in leveraging hardware resources of GPUs but also in scaling with the number of users, and maintaining service quality while controlling costs.

\begin{figure}[t]
\vspace{-0.3cm}
    \centering
    \includegraphics[width=0.865\linewidth]{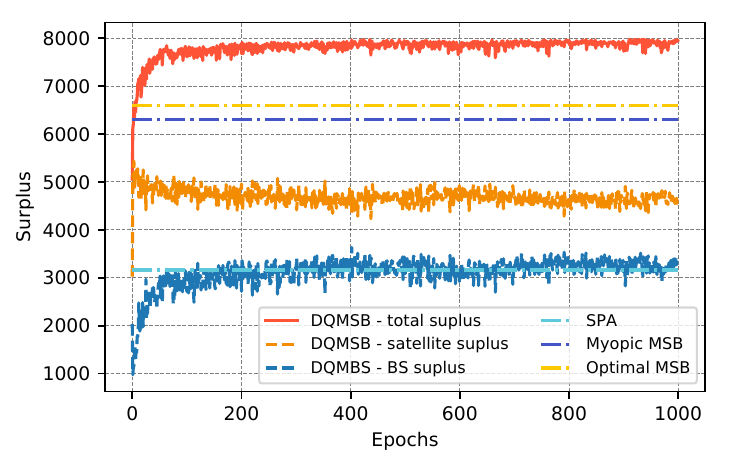}
    \caption{The convergence of the proposed DQMSB auction.}
    \label{fig:convergence}
\end{figure}

\begin{figure*}[t]
    \centering
    \subfigure[Total surplus versus number of ground BSs.]{\includegraphics[width=0.24\textwidth]{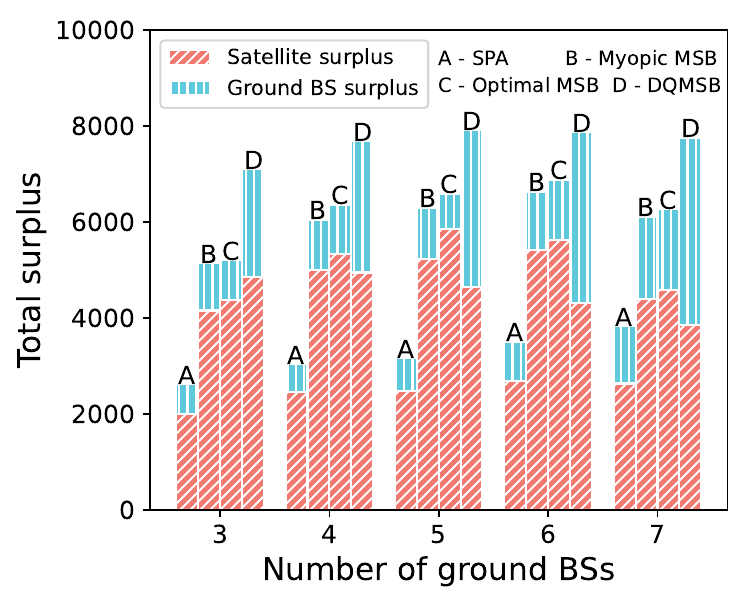}%
        \label{fig:auctionBS}}
    % \hfil
    \subfigure[Total surplus versus number of services.]{\includegraphics[width=0.24\textwidth]{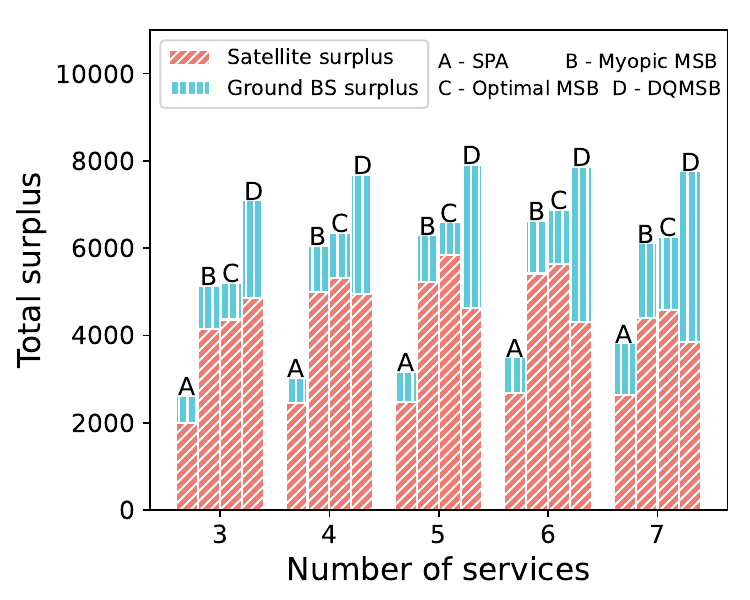}%
        \label{fig:auctionservices}}
    \subfigure[Total surplus versus number of GPUs.]{\includegraphics[width=0.24\textwidth]{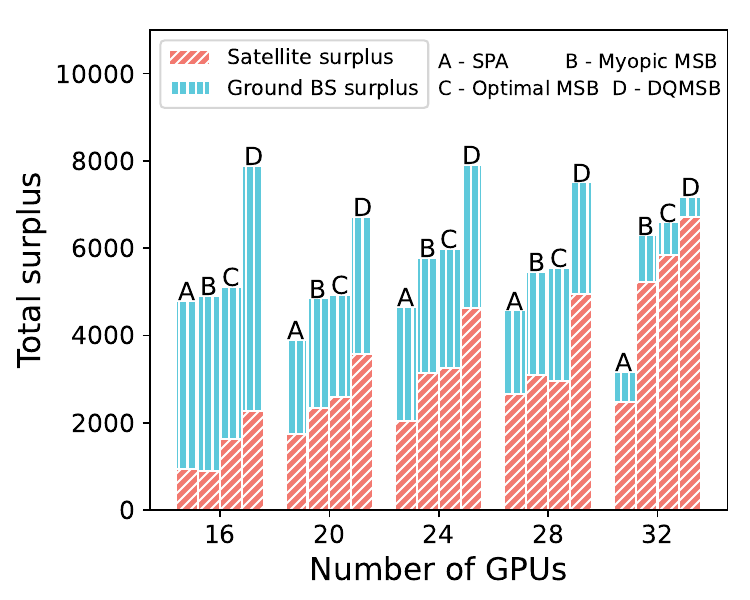}%
        \label{fig:auctionGPUs}}
    \subfigure[Total surplus versus number of users.]{\includegraphics[width=0.24\textwidth]{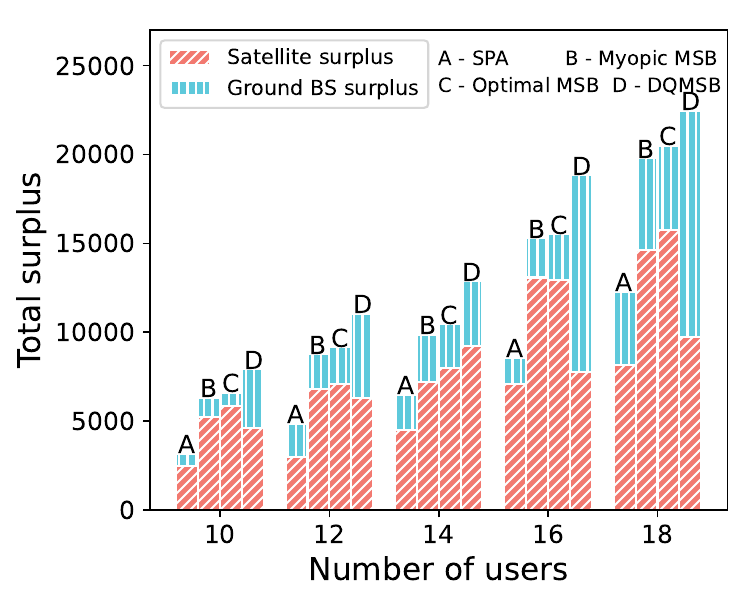}%
        \label{fig:auctionusers}}
    \caption{Performance of the proposed DQMSB auction under different system settings.}
    \label{fig:auction}
\end{figure*}

As shown in Fig.~\ref{fig:accuracycost}, as the context factor increases, the cost of accuracy also increases for all three algorithms (FIFO, LFU, and least AoT). This indicates that when the context becomes more influential, which can be seen as the importance of cache relevance, it becomes more challenging and expensive for all methods to maintain high accuracy in caching decisions. Meanwhile, in Fig.~\ref{fig:performancegain}, there is a decrease in performance gain for all algorithms as the context factor increases. This suggests that as the importance of the context in caching decisions increases, the performance gain of these algorithms decreases. The overall decrease in performance gain with a higher context factor suggests that as the relevance of the context increases, the potential for any algorithm to outperform basic caching strategies diminishes. This might be because the decision-making process becomes more complex and the benefits of sophisticated strategies are less pronounced. % In general, FIFO appears to be the least effective in this aspect, while the least AoT algorithm is the most resilient against changes in context, maintaining a lower cost of accuracy and higher performance gains.

\subsection{Convergence Analysis}
Initially, we demonstrate the convergence performance of the DQMSB auction in Fig.~\ref{fig:convergence}. At the beginning of training, the total surplus achieved by the DQMSB auction starts with a sharp increase and then levels off, indicating that the mechanism quickly learns an effective strategy for maximizing surplus and then converges to a stable solution. At around 200 epochs, the total surplus stabilizes at a high level. Although there are minor fluctuations following this rise, the surplus remains relatively consistent, indicating that the system has reached a convergence in its learning phase. Interestingly, the performance of DQMBS in ground BS surplus can achieve similar performance with the SPA, at around 3,000. This indicates that the SPA can realize the surplus of satellites, which can reach around 4,500 for the optimal solution. Furthermore, the total surplus achieved by the DQMSB auction can outperform the MSB auction by around 20\%. The convergence performance of the DQMSB auction can be considered constantly robust, as it achieves and maintains a higher total surplus compared to the benchmarks. % This suggests that the DQMSB auction is effective at learning and applying the optimal price scaling factor for allocating network operators for provisioning LLM agent services.

\subsection{Performance Evaluation for the DQMSB Auction}

To evaluate the performance of the DQMSB auction, we leverage several auction baselines, including second-price auction, myopic MSB, and optimal MSB. Particularly, the price scaling factor of myopic MSB is set as $\rho = \max(1, x_0/x_{(2)})$ with current round information and the price scaling factor of optimal MSB can be set as $\rho = \max(1, \mathbb{E}[x_0]/\mathbb{E}[x_{(2)}])$ with historical statistic information. Under different numbers of bidders, Fig. ~\ref{fig:auctionBS} demonstrates the total surplus achieved by various auction mechanisms based on the number of ground BSs. Across all auction mechanisms, the total surplus increases as the number of ground BSs grows from 3 to 7. This suggests that a larger number of ground BSs enables better service provision and coverage, resulting in a higher overall value generated by the auctions. The DQMSB auction consistently yields the highest total surplus, regardless of the number of ground BSs. This indicates that the DQMSB auction is more efficient in resource allocation and surplus generation compared to other auction types. The Myopic MSB and Optimal MSB auctions outperform SPA but fall short of the performance achieved by the DQMSB auction. % It is worth noting that the Optimal MSB auction approaches the performance of DQMSB, particularly as the number of BSs increases. The increasing trend of surplus with the number of BSs suggests scalability in the network services market, as the infrastructure expands, leading to a growth in the total market surplus.
As we can observe from Fig. \ref{fig:auctionservices}, the total surplus does not monotonically increase or decrease with the number of services. For instance, there is a drop in total surplus for all mechanisms when varying from 10 to 12 services, followed by an increase at 14 services, another decrease at 16, and an increase again at 18. In most cases, the surplus from ground BS services is greater than the surplus from satellite services. 
% The exception is at 10 services, where the satellite surplus is higher for all the mechanisms.
This non-linear relationship implies that simply increasing the number of services cannot guarantee a higher surplus. % There may be an optimal range of service offerings that maximizes the total surplus, which could be subject to diminishing returns or market saturation effects.

In Fig.~\ref{fig:auctionGPUs}, there appears to be a general increase in the total surplus for all the auctions as the number of GPUs increases. The surplus for the other mechanisms also tends to increase, although not as consistently or significantly as the DQMSB auction. The ground BS surplus dominates the total surplus for all auction mechanisms and quantities of GPUs. However, as the number of GPUs increases, the satellite surplus also increases, suggesting a positive relationship between computing resources and the ability to generate surplus in satellite-based services. The increasing trend of total surplus with more GPUs implies that having more computing resources allows the auction mechanisms, particularly DQMSB, to perform better. % This is attributed to more sophisticated bidding strategies or improved learning of price scaling factors facilitated by increased computing power. The results indicate that investing in computing resources of GPUs can be justified by the corresponding increase in total surplus, especially for the DQMSB auction.
Finally, Fig.~\ref{fig:auctionusers} demonstrates an upward trend in the total surplus with an increasing number of users for all auction mechanisms, which suggests that more users contribute to a higher valuation and competition, thus increasing the total surplus. When the number of users increases, the DQMSB auction can yield a larger surplus from ground BSs, thereby limited computing and communication resources can be allocated effectively to maximize the total surplus.

% \begin{figure}[t]
%     \centering
%     \includegraphics[width=0.8\linewidth]{Figs/auction_compare_GPUs.pdf}
%     \caption{Total surplus versus number of GPUs.}
%     \label{fig:auctionGPUs}
% \end{figure}

% \begin{figure}[t]
%     \centering
%     \includegraphics[width=0.8\linewidth]{Figs/auction_compare_users.pdf}
%     \caption{Total surplus versus number of users.}
%     \label{fig:auctionusers}
% \end{figure}
\section{Conclusions}\label{sec:conclusions}

In this paper, we proposed a joint caching and inference framework for provisioning ubiquitous edge intelligence services in SAGINs. In SAGINs, satellites and ground BSs were utilized to provision global LLM agent services with edge servers at ground BSs or remote cloud data centers. Specifically, considering the unique few-shot learning capabilities of LLMs and new constraints on the size of context windows, we introduced a new concept, i.e., the cached model as a resource, beyond conventional communication, computing, and storage resources. For allocating cached model resources, we designed a new metric, namely, age of thought, to evaluate the relevance and consistency of thoughts/CoT examples in context windows during inferences and proposed the least AoT algorithm. Finally, we proposed the DQMSB for incentivizing network operators to provision LLM agents with high market efficiency through the DQN-based price scaling factor. Theoretically, we proved that the proposed DQMSB auction is anonymous, fully strategy-proof, and adverse-selection-free. 
% \end{comment}
\bibliographystyle{IEEEtran}
\bibliography{main}

\end{document}